\documentclass[12pt]{article}

\usepackage[utf8]{inputenc}
\usepackage[T1]{fontenc}
\usepackage[margin=1in]{geometry}

\title{A Latent Factor Panel Approach to Spatiotemporal Causal Inference}

\author{
Jiaxi Wu\thanks{This work was completed independently and does not relate to the author's position at Amazon.}\\
University of California, Santa Barbara
\and
Alexander Franks\\
University of California, Santa Barbara
}

\usepackage{mathtools}
\usepackage{amsmath} 
\usepackage{amssymb}
\usepackage{amsthm}
\usepackage{bm}
\usepackage{graphicx}
\usepackage{caption}
\usepackage{subcaption}
\usepackage{booktabs}
\usepackage{enumitem}
\usepackage{hhline}
\usepackage{wrapfig}
\usepackage{algorithm}
\usepackage{algorithmic}

\usepackage{xurl}

\usepackage{natbib}
\newtheorem{theorem}{Theorem}

\newtheorem{proposition}{Proposition}

\newtheorem{assumption}{Assumption}

\newcommand{\indep}{\perp \!\!\! \perp}

\newcommand{\Cov}{\operatorname{Cov}}
\newcommand{\E}{\mathbb{E}}

\newcommand{\pmtf}{$\text{PM}_{2.5}$ }

\usepackage{hyperref}
\hypersetup{
    colorlinks=true,
    linkcolor=blue,
    citecolor=blue,
    urlcolor=blue,
    pdftitle={A Latent Factor Panel Approach to Spatiotemporal Causal Inference},
    pdfauthor={Jiaxi Wu and Alexander Franks}
}

\begin{document}

\maketitle


\begin{abstract}
Unmeasured confounding can severely bias causal effect estimates from spatiotemporal observational data, especially when the confounders do not vary smoothly in time and space. In this work, we develop a method for addressing unmeasured confounding in spatiotemporal contexts by building on models from the panel data literature and methods in multivariate causal inference.  Our method is based on a \emph{factor confounding} assumption, which posits that effects of unmeasured confounders on exposures and outcomes can be captured by a shared latent factor model. Factor confounding is sufficient to partially identify causal effects, even when there is interference between units. Additional assumptions that limit the degree of spatiotemporal interference, reasonable in most applications, are sufficient to point identify the effects. Simulation studies demonstrate that the proposed approach can substantially reduce omitted variable bias relative to other spatial smoothing and panel data baselines.  We illustrate our method in a case study of the effect of prenatal $\text{PM}_{2.5}$ exposure on birth weight in California.


\vspace{1.5em}

\end{abstract}

\clearpage

\section{Introduction}


Causal inference with spatial data is essential in disciplines where exposures are determined by geography rather than experimental design---for instance, environmental exposures in epidemiology \citep{bind2019causal} and neighborhood socioeconomic conditions in public health research \citep{diez2001investigating}. In such contexts, researchers cannot randomize who is exposed to polluted air, which census tract receives a policy intervention, or which areas are prioritized for infrastructure development. Instead, they must draw causal insights from observational data in which sampling units are embedded in space. The spatial configuration presents both opportunities---through natural variation and spatial dependence---and challenges, as exposures, covariates, and outcomes often co-vary spatially. However, as in any observational study, the validity of causal conclusions hinges on untestable assumptions about the confounding mechanism. The term \textit{spatial confounding} has been used to describe a number of phenomena arising from either the data generation process or the analysis model \citep{hodges2010adding, khan2023re}. Following the causal perspective of \citet{gilbert2021causal}, we focus on unmeasured confounders that simultaneously influence both the exposure and the outcome. 


Recent work in spatial statistics has examined the conditions under which different spatial and causal methods account for spatial confounding \citep{reich2021review}. For example, \citet{thaden2018structural} and \citet{dupont2022spatial+} employ two-stage procedures to partial out spatial confounding. \citet{papadogeorgou2019adjusting} incorporate spatial proximity into propensity score matching to adjust for unmeasured confounding, while \citet{schnell2020mitigating} model the joint distribution of the exposure and unmeasured confounders parametrically to mitigate bias. \citet{papadogeorgou2023spatial} propose a Bayesian parametric approach to account for entangled spatial confounding and interference. 
Most existing approaches depend on parametric assumptions regarding the relationships between unobserved confounders and observed variables, often focusing on specific estimands or analysis frameworks. In contrast, \citet{gilbert2021causal} formalize a set of nonparametric identification assumptions for geostatistical (point-referenced) data, requiring the unmeasured confounder to be a fixed, measurable function of spatial location.

A common premise behind many adjustments for spatial confounding is that the exposure varies at a finer spatial scale than the spatial confounder \citep{paciorek2010importance, keller2020selecting, guan2023spectral}, or exhibits sufficient non-spatial variation to satisfy the positivity assumption required for causal identification \citep{bobb2022accounting, schnell2020mitigating, gilbert2021causal}. Scale-separation approaches operationalize this premise by decomposing spatial signals into frequency bands and filtering low-frequency confounding. In settings with multiple exposures and outcomes, \citet{prim2025spectral} develop a spectral-domain adjustment by projecting to scale-specific components under a local unconfoundedness condition. In practice, however, unmeasured confounders may vary over space, fluctuate over time, and display idiosyncratic differences across units, leading to mixed spatial and non-spatial structure. Such complexity limits the effectiveness of existing methods that rely solely on spatial (or spatiotemporal) smoothing to adjust for confounding bias.

To account for some of these complexities, we leverage ideas from the panel data literature.  Causal panel methods are designed to handle unobserved heterogeneity across multiple units measured repeatedly, most commonly applied to settings with binary exposures and well-defined pre- and post-treatment periods.  Popular approaches include classical two-way fixed effects, difference-in-differences, and synthetic control or factor-based methods \citep[see e.g.][for a review]{ben2021trial, arkhangelsky2024causal}. Of particular relevance is the interactive fixed effects (IFE) model of \citet{bai2009panel}, which represents unobserved heterogeneity through a low-rank decomposition of time-varying latent factors and unit-specific loadings. This method accommodates continuous exposures and yields asymptotically consistent estimates under linear outcome models, provided that both the cross-sectional and temporal dimensions are sufficiently large.


Our approach also builds on recent advances in causal inference with multiple exposures \citep{miao2022identifying, zheng2025copula} and multiple outcomes \citep{zheng2024sensitivity, kang2023partial}. With appropriate assumptions, causal effects can be partially identified, underscoring the importance of linking identification assumptions to causal conclusions in a transparent and interpretable manner \citep{zheng2022bayesian}. One common strategy leverages shared confounding via the use of proxy variables or negative controls to identify causal effects or detect bias \citep{shi2020selective, tchetgen2024introduction}. These methods rely on variables that are assumed to be associated with the unmeasured confounder but are known not to directly influence the outcome (negative control exposures) or not to be affected by the exposure (negative control outcomes). While negative controls have been explored in the context of time-series and longitudinal studies, their application to spatial settings remains limited, with some exceptions \citep[e.g.][]{lumley2000assessing}. 


\subsection{Contribution}


In this work, we develop a flexible framework for reasoning about causal effects in spatiotemporal settings with unmeasured confounding and potential interference.  \citet{gilbert2021causal} show that it is possible to identify causal effects from spatial data as long as unmeasured confounders are a deterministic measurable function of the spatial location.  However, even when such assumptions hold in principle, reliable estimation can fail in practice when the number of sampled sites is small or in applications with short-range correlations or non-stationarity. In order to handle such situations, we apply recent results related to multivariate causal inference with latent variables to account for residual spatiotemporal variation. Previous studies have shown that when there are multiple measurements on a unit sharing a common confounding mechanism, latent variable models can facilitate (partial) identification by leveraging the shared information---typically across multiple exposures or outcomes within each unit \citep{zheng2024sensitivity, zheng2025copula, kang2023partial}. We extend this idea of exploiting \emph{within-unit} associations to modeling \emph{between-unit} dependencies, where latent confounders induce correlations across space and/or time points. 

We show that when exposures and outcomes are linear in the unmeasured confounders, with appropriate conditions, factor models can facilitate partial identification of causal effects. We then establish precise limits on the degree of spatiotemporal interference that can be incorporated while maintaining causal identification. In addition, our approach is compatible with a wide range of models, including those that allow for treatment effect heterogeneity. We provide theoretical guarantees, conduct extensive simulations benchmarking our approach against commonly used methods, and demonstrate its practical utility through a case study estimating the effect of PM$_{2.5}$ exposure on birth weight in California.
%
%





\section{Problem Setup}
\label{sec:setting}

%

Let $\mathcal I=\{1,\dots,N\}$ index spatial locations and $\mathcal T=\{1,\dots,T\}$ index time points. For each location–time pair $(i,t)\in\mathcal I\times\mathcal T$, we observe $(D_{it},\,Y_{it},\,\mathbf X_{it},\,\mathbf S_i)$ drawn from a joint spatiotemporal data generating process $\mathcal{P}$, where $D_{it}\in\mathbb R$ is the exposure (or treatment), $Y_{it} \in \mathbb{R}$ is the outcome, $\mathbf{X}_{it}\in \mathbb{R}^p$ is a vector of observed pre-exposure covariates, which may include both unit-specific covariates and observed temporal features. $\mathbf{S}_i \in \mathbb{R}^{p_s}$ denotes the spatial coordinates of location $i$ (e.g., $p_s=2$ for longitude and latitude). Collect the exposures and outcomes into $\mathbf D= (D_{it}) \in\mathbb R^{N\times T}$ and $\mathbf Y=(Y_{it})\in\mathbb R^{N\times T}$, with $\mathbf D_t$ and $\mathbf Y_t$ denoting the $N$-vectors at time $t$, and $\mathbf D_i$ and $\mathbf Y_i$ denoting the $T$-vectors at location $i$.

Let $Y_{it}(\mathbf d)$ be the potential outcome that would have been observed at $(i,t)$ had all exposures been set to $\mathbf d$ \citep{rubin1974estimating}. The dimensionality of potential outcomes can be reduced under the partial interference assumption \citep{reich2021review, papadogeorgou2023spatial}. Specifically, we assume that for each $(i,t)$, $Y_{it}$ depends only on exposures in the neighborhood $\mathcal N_{it} \subseteq \mathcal I\times\mathcal T$:
 
\begin{assumption}[Partial interference]
\label{asm:partial_inter}
For every $(i,t)$ and exposure $\mathbf d$, $Y_{it}(\mathbf d)=Y_{it}(\mathbf d_{\mathcal N_{it}})$, where $\mathbf d_{\mathcal N_{it}} :=\{d_{jk}:(j,k)\in\mathcal N_{it}\} \in\mathbb R^{|\mathcal N_{it}|}$.
\end{assumption}

Common choices include: no interference, where $\mathcal N_{it}=\{(i,t)\}$; spatial neighbors, where $\mathcal N_{it} = \{(j,t): j \in \mathcal N_i \subseteq \mathcal I\}$; or temporal interference such that $\mathcal N_{it} = \{(i,k): k \leq t\}$. It is common to further summarize the exposure through a measurable map $\mathbf A_{it} = f_{it}(\mathbf D_{\mathcal N_{it}})$, for instance $\mathbf A_{it}=(D_{it},\bar D_{it})$ with $\bar D_{it}$ the neighborhood mean \citep{aronow2017estimating}. The potential outcome can be written as $ Y_{it}(\mathbf d) = Y_{it}(\mathbf d_{\mathcal N_{it}}) = Y_{it}(\mathbf a)$. The population average treatment effect (PATE) comparing two interventions $\mathbf d^{(1)}$ and $\mathbf d^{(2)}$ is defined as
\[
   \operatorname{PATE}(\mathbf d^{(1)},\mathbf d^{(2)}) \coloneqq
   \E \left[
       Y_{it}\bigl(\mathbf d^{(1)}_{\mathcal N_{it}}\bigr)
       -Y_{it}\bigl(\mathbf d^{(2)}_{\mathcal N_{it}}\bigr)
   \right],
\]
where the expectation is taken over the population distribution $\mathcal{P}$. This estimand includes, as special cases, the direct effect, where the two exposure vectors differ only in the exposure of the index unit $(i,t)$, with all neighbors' exposures held fixed, and the indirect (interference) effect, where they differ only in the exposures of neighbors, while keeping the index unit's exposure fixed.


The unmeasured confounding is represented by the latent vector $\mathbf U_{it}$. If $\mathbf U_{it}$ were observed, the following assumptions would suffice to identify the causal effect nonparametrically:

\begin{assumption}[Latent positivity]
\label{asm:positivity}
$f_{\mathbf D_{\mathcal N_{it}}\mid \mathbf X_{it},\mathbf S_i,\mathbf U_{it}} (\mathbf d\mid \mathbf x, \mathbf s, \mathbf u) > 0$ for every $(\mathbf d, \mathbf x, \mathbf s, \mathbf u)$ in the support.
\end{assumption}

\begin{assumption}[Latent unconfoundedness]
\label{asm:unconfoundedness}
$Y_{it}(\mathbf d_{\mathcal N_{it}}) \indep \mathbf D_{\mathcal N_{it}} \mid (\mathbf X_{it},\,\mathbf S_i,\,\mathbf U_{it})$ for all $\mathbf d_{\mathcal N_{it}}$.
\end{assumption}

\noindent Assumptions \eqref{asm:partial_inter}-\eqref{asm:unconfoundedness} are relaxations of the standard exchangeability, positivity and SUTVA assumptions that are commonly made to identify causal effects from observational data. Under Assumptions~\ref{asm:partial_inter}–\ref{asm:unconfoundedness}, the expectation of the potential outcome is
\[
   \mathbb E\bigl[Y_{it}(\mathbf d_{\mathcal N_{it}})\bigr]
   \;=\;
   \mathbb E_{\mathbf X,\mathbf S,\mathbf U} [
       \mathbb E[Y_{it} \mid \mathbf D_{\mathcal N_{it}}=\mathbf d_{\mathcal N_{it}}, \mathbf X_{it},\mathbf S_i,\mathbf U_{it}]
   ],
\]
and any average causal contrast such as the PATE between two vectors $\mathbf d^{(1)}_{\mathcal N_{it}}$ and $\mathbf d^{(2)}_{\mathcal N_{it}}$ is identified.

Since $\mathbf U_{it}$ is unobserved, adjustment is limited to the observed covariates $(\mathbf X_{it},\mathbf S_i)$ and causal effects are not identifiable without additional assumptions. In spatial analyses, spatial smoothness assumptions provides a crucial bridge between observed data and unmeasured confounders. To mitigate confounding bias, a common approach is to model unmeasured confounding as a spatially varying confounder, assuming that its variability is fully captured by spatial location. However, in practice, $\mathbf U$ often exhibits both spatial and non-spatial variation, which makes adjustments based solely on $\mathbf X$ and $\mathbf S$ insufficient for identifying causal effects. We adopt this perspective throughout the paper. The confounding bias represents the discrepancy between the naive estimate under the assumption of no unmeasured confounding (NUC) and the true causal effect. For PATE, it can be expressed as
\begin{align*}
    \text{Bias}(\mathbf d^{(1)},\mathbf d^{(2)}) &= \mathbb E_{\mathbf X, \mathbf S}\bigl[
    \mathbb E[Y_{it}\mid \mathbf D_{\mathcal N_{it}} =\mathbf d^{(1)}_{\mathcal N_{it}},\mathbf X_{it},\mathbf S_i ] - \mathbb E[Y_{it}\mid\mathbf D_{\mathcal N_{it}} =\mathbf d^{(2)}_{\mathcal N_{it}},\mathbf X_{it},\mathbf S_i]
    \bigr] \\ &\quad - \operatorname{PATE}(\mathbf d^{(1)},\mathbf d^{(2)}).
\end{align*}


%
\noindent Next, we introduce our factor confounding model along with the assumptions under which the bias above becomes identifiable.

\section{Factor Confounding}
\label{sec:shared}




In this section, we formalize how additional assumptions about low-rank confounding, combined with factor models and negative-control-style assumptions, can help debias causal effect estimates. Throughout, we assume that, given all observables, the effect of unmeasured confounding on both exposures and outcomes is reflected in the residual dependence structure, which we model as a rank-$M$ matrix plus a diagonal component.  For this reason, we refer to our approach as the ``factor confounding'' model. Below, in Section \ref{sec:model_spec}, we start by outlining the core modeling assumptions of the factor confounding model used to account for residual spatial correlation. We provide identification results in Section \ref{sec:identification} and in Section \ref{sec:correlation_comparison}, we briefly discuss tradeoffs between modeling residual spatial vs residual temporal correlation. Finally, we describe our estimation strategy in Section \ref{sec:estimation}.

\subsection{Model Specification}
\label{sec:model_spec}

Here, we present our factor confounding model in a setting where unmeasured confounding is captured through residual spatial variation, with time points treated as independent replicates. This formulation yields a tractable representation of the bias and forms the basis for the identification results developed below. Later, we discuss an alternative formulation based on residual temporal correlation. Localized neighboring or lagged exposures can likewise be incorporated through the specification of the exposure neighborhood in the outcome model. Below, we state a set of assumptions which are sufficient for identifying causal effects in the presence of unmeasured confounding. In Section \ref{sec:sim_misspec}, we show that the results are robust to several practically relevant violations of these assumptions.

\begin{assumption}[Factor confounding]
\label{asm:linear_U}
Conditional on $(\mathbf X_{it},\mathbf S_i)$, both $\mathbf{D_{t}}$ and $\mathbf{Y_{t}}$ are linear functions of an unmeasured $\mathbf U_{t} \in\mathbb{R}^{M}$:
\begin{align} 
\mathbf{D}_t &= f(\mathbf{X_t}, \mathbf{S}) + B\mathbf U_{t} + \bm{\xi}_t, \label{eqn:factor_d}\\
\mathbf{Y}_t &= g(\mathbf{D}_t, \mathbf{X_t}, \mathbf{S}) + \Gamma\Sigma_{U\mid D}^{-1/2} \mathbf U_{t} + \bm{\varepsilon}_t, \label{eqn:factor_y} 
\end{align}
where we denote
$\mu_{U\mid D} := \E[\mathbf U_t\mid \mathbf{D}_t, \mathbf X_t, \mathbf{S}]$ and 
$\Sigma_{U\mid D} := \Cov(\mathbf U_t\mid \mathbf{D}_t, \mathbf X_t, \mathbf{S})$.
Marginally, $\mathbf U_{t}$ are i.i.d.\ mean-zero $M$-vectors. $\bm{\xi}_{t}$ and $\bm{\varepsilon}_{t}$ are independent mean-zero random vectors with diagonal covariance $\Lambda_D$ and $\Lambda_Y$ respectively. We assume that $\mathbf{D_t}$ does not affect future confounders, so that there is no treatment-confounder feedback \citep{hernan2020causal}. 

\end{assumption}


\noindent 


\noindent Under this model, $\Cov(vec(\mathbf D) \mid \mathbf X, \mathbf S) = I_T \otimes (B B^{\top} + \Lambda_D)$ and $Cov(vec(\mathbf Y) \mid \mathbf D, \mathbf X, \mathbf S) = I_T \otimes (\Gamma \Gamma^{\top} + \Lambda_Y)$ where $B,\Gamma \in \mathbb R^{N\times M}$ are factor loading matrices. Any causal estimand can be written as a function of $g(\mathbf{D}_t) := (g_1(\mathbf{D}_t), \dots, g_N(\mathbf{D}_t)^\top$. 

To streamline notation moving forward, we suppress explicit conditioning on the observed covariates and spatial locations, $(\mathbf{X}, \mathbf{S})$, but all results hold conditionally on them. The bias arises from the naive regression of $\mathbf{Y}_t$ on $\mathbf{D}_t$ is given by the following proposition.


\begin{proposition}[Bias of naive regression]\label{prop:naive-bias}
Under model~\eqref{eqn:factor_y}, the naive regression of the outcome vector $\mathbf{Y}_t$ on the exposure vector $\mathbf{D}_t$ gives the conditional moments
\begin{align}
    \E[\mathbf Y_t \mid \mathbf D_t]
        &= g(\mathbf D_t)+
           \Gamma\Sigma_{U\mid D}^{-1/2}
           \mu_{U\mid D}, \label{eqn:outcome_mean}\\
        \Cov(\mathbf Y_t \mid \mathbf D_t)
        &= \Gamma\Gamma^{\top} + \Lambda_Y. \label{eqn:outcome_var}
\end{align}
Consequently, the naive estimator of $g$ is additively biased by $\text{Bias}(\mathbf{d}) = \Gamma\Sigma_{U\mid D}^{-1/2} \mu_{U\mid D}$.
\end{proposition}

\noindent This proposition highlights the parametric link between the outcome mean and covariance.  Under this model, exact recovery of $g$ is only possible if the bias term is zero, or under additional assumptions on the magnitude or structure of the causal effects. 

To complete the model specification, we need to specify a latent variable model for the exposures that leads to potentially identifiable expressions for $\mu_{U\mid D}$ and $\Sigma_{U\mid D}$. While many choices are possible, we focus here on the Gaussian probabilistic PCA formulation \citep{tipping1999probabilistic}.
\begin{assumption}[Probabilistic PCA exposure model]
Exposures are linear in Gaussian unmeasured confounders with additive Gaussian noise:
\begin{align}
\mathbf{D}_t &= B \mathbf U_{t} + \bm{\xi}_t, \label{eqn:exposure_factor}\\
\mathbf U_{t} &\sim \mathcal{N}_M(0, I_M), \\
\bm{\xi}_t  &\sim \mathcal{N}_N(0, \Lambda_D),  \label{eqn:exposure_residual}
\end{align}
where $\Lambda_D$ is a diagonal covariance matrix. 
\label{asm:ppca}
\end{assumption}

\noindent Under this specification, the exposure distribution is $\mathbf{D}_t \sim \mathcal N_N(0, \Sigma_{D})$ with $\Sigma_{D}= BB^{\top}+\Lambda_D$. The joint normality gives $\mathbf U_{t}\mid \mathbf{D}_t \sim \mathcal{N}_{M}(\mu_{U\mid D} , \Sigma_{U\mid D})$, where
$\mu_{U\mid D} = B^\top\Sigma_D^{-1}\mathbf{D}_t$ and 
$\Sigma_{U\mid D} = I_M - B^\top\Sigma_D^{-1}B. $\label{eqn:var_u_given_d}

In practice, we find that our method is robust to violations of both Assumption \ref{asm:linear_U} and \ref{asm:ppca}. In particular, we show via simulation that we can recover nearly unbiased causal effects when  $\mathbf U$ is non-Gaussian or when the exposures and outcomes are nonlinear functions of $\mathbf U$ (see Section \ref{sec:sim_misspec}). There are several possible alternative approaches for relaxing Assumption \ref{asm:ppca}, but we do not explore them further here \citep[see e.g.][]{zheng2024sensitivity, zheng2025copula}.

Finally, we note that our method is closely related to the interactive fixed effects (IFE) estimator proposed by \citet{bai2009panel}, which extends traditional fixed effects models by introducing low-rank time-unit interactions to capture unobserved heterogeneity.
%
%
Like our approach, the IFE model uses factor structure in both exposures and outcomes to account for heterogeneity and to de-bias estimates of exposure effects. However, there are crucial differences in both motivation and structure. First, \citet{bai2009panel} focus on parameter identification in an asymptotic framework, and do not explicitly address causal identifiability or omitted variable bias, nor do they account for the role of potential interference. Second, they treat all parameters as fixed effects, whereas we consider the implications of a model where shared confounders are random effects that can be marginalized out. This leads to significant performance gains over IFE when either $N$ or $T$ is small, as we show via simulations. Moreover, the IFE model requires outcomes to be linear in the exposures, while our approach can accommodate both nonlinear exposure effects and treatment effect heterogeneity across units. Below, we show how our formulation leads to explicit (partial) identification of causal parameters without relying on asymptotic arguments. We then highlight advantages of our model through simulations in Section \ref{sec:simulation} and the application in Section \ref{sec:application}.

\subsection{Identification under Factor Confounding}
\label{sec:identification}

To bound causal effects even in the presence of unrestricted interference, we must at minimum be able to identify the factor loadings in both the exposure model and the outcome model. The following standard assumption guarantees identifiability of the factor loadings up to rotation.

\begin{assumption}[Factor model identifiability]
\label{asm:factor}
The outcomes follow Equation~\eqref{eqn:factor_y} and the exposures satisfy \eqref{eqn:exposure_factor}-\eqref{eqn:exposure_residual}. Further, the factor loading matrices $B$ and $\Gamma$ are of rank $M$, and for each matrix, removing any row results in two disjoint submatrices, each of full rank $M$. Under this condition, $B$ and $\Gamma$ are identifiable up to rotations from the right.
\end{assumption}

This assumption typically requires $(N-M)^2-N-M\geq 0$ and each latent confounder affects at least three exposures and outcomes \citep{anderson1956statistical}. Under Assumption~\ref{asm:factor}, we obtain bounds on the omitted variable bias.

\begin{proposition}
Under model~\eqref{eqn:factor_y} and Assumptions~\ref{asm:ppca} - \ref{asm:factor} the causal effect function, $g(\cdot)$, is partially identified. Let $\check \Gamma$ and $\check B$ be any matrices such that $\check \Gamma \check \Gamma^{\top} = \Gamma \Gamma ^{\top}$ and $\check B \check B^{\top} = BB^\top$. Let $\check \gamma_i$ denote the $i$th row of $\check \Gamma$. Then the omitted variable bias for the causal effect of the exposure vector $\mathbf d$ on the outcome at site $i$ is
\begin{equation}
\text{Bias}(\mathbf{d})_i = \check \gamma_i \Theta \check \Sigma_{U \mid D}^{-1/2}\check B^{\top}\Sigma_{D}^{-1}\mathbf d \in \pm \parallel \check \gamma_i\parallel_2 \parallel\check \Sigma_{U \mid D}^{-1/2}\check B^{\top}\Sigma_{D}^{-1}\mathbf d\parallel_2
\label{eqn:bias_rotation}
\end{equation}
\noindent where $\Theta \in \mathcal{O}_M$ is an $M \times M$ orthogonal matrix.  The interval on the right-hand side of \eqref{eqn:bias_rotation} is identifiable for all $i$. $\Theta$, and hence $g(\cdot)$, are not identifiable without further assumptions.  
\label{prop:bias_bound}\end{proposition}
\noindent Proposition \ref{prop:bias_bound} shows that causal effects are only partially identified under the factor confounding assumption. Additional structural or spatial assumptions are needed to point identify the effects. In this work, we focus on situations with limited neighborhood interference. Specifically, we assume each outcome $Y_{it}$ depends only on exposures within a given neighborhood, that is $g_i(\mathbf{D}_t) = g_i\bigl( (D_{jt})_{j\in \mathcal N_i} \bigr)$.
Under this assumption, exposures at locations $j \notin \mathcal N_i$ act as negative controls for $Y_{it}$, that is, $\frac{\partial g_i(\mathbf{D}_t)}{\partial D_{jt}} = 0$ for all $\mathbf{D}_t$. In the special case where there is no interference, each outcome $Y_{it}$ depends only on its own local exposure $D_{it}$, i.e.  $g_i(\mathbf{D}_t) = g_i(D_{it})$. Next, we introduce a formal assumption that enables identification of causal effects, even under neighborhood interference.
%
%
%
%
%
\begin{assumption}[Off-neighborhood rank]\label{asm:rank}
Let $\gamma_i^{\top}$ denote row $i$ of $\Gamma$ and $r_j$ denote column $j$ of matrix $R:=\Sigma_{U \mid D}^{-1/2}B^{\top}\Sigma_{D}^{-1} \in\mathbb R^{M\times N}$. For each unit $i$, define its interference neighborhood $\mathcal N_i\subseteq \mathcal I$ as the set of spatial locations $j$  satisfying $\frac{\partial g_i(\mathbf{D}_t)}{\partial D_{jt}} \neq 0$ for all $\mathbf{D}_t$. Then, there exist $M$ distinct indices
$I_\star:=\{i_1,\dots ,i_M\}\subseteq \mathcal I$ such that
\begin{enumerate}
\item $\{\gamma_{i_\ell}\}_{\ell=1}^M$ form a basis of $\mathbb R^{M}$;
\item for each $\ell=1,\dots ,M$, the submatrix
      \(
          R_{\,J(i_\ell)} :=\bigl(r_j\bigr)_{j\notin\mathcal N_{i_\ell}^{}} \in\mathbb R^{M\times(N-|\mathcal N_{i_\ell}^{}|)}
      \)
      has full row rank $M$.
\end{enumerate}
\end{assumption}
\noindent Intuitively, this assumption guarantees that at least $M$ units have enough exposures outside their own interference neighborhoods to span all $M$ latent confounding directions. In many applications, the number of latent confounders $M$ and the size of each neighborhood $\lvert\mathcal{N}_i\rvert$ are both much smaller than the number of spatial units $N$, so Assumption \ref{asm:rank} will hold. This assumption leads to the following identification result.






\begin{theorem}
\label{thm:identifiability_shared_confounding}
Under model~\eqref{eqn:factor_y} and Assumptions~\ref{asm:ppca} - \ref{asm:rank}, the causal effect functions, $g_i\bigl( (D_{jt})_{j\in \mathcal N_i}\bigl)$, are identified for all $i$.


  
\end{theorem}

\noindent We briefly provide intuition for this result. Define the rank-$M$ bias matrix as $C \;:=\; \Gamma \Sigma_{U\mid D}^{-1/2} B^{\top}\Sigma_{D}^{-1}$.  Then we have conditional expectation $\E[Y_{it} \mid \mathbf{D_t}] = g_i((D_{jt})_{j\in \mathcal N_i}\bigl) + (C\mathbf{D_t})_i$.  For any $j \notin \mathcal{N}_i$, any observed association between $Y_{it}$ and $D_{jt}$ given all other exposures must be due to shared confounding, which then point identifies $C_{ij}$.  Assumption \ref{asm:rank} ensures that enough entries of $C$ are identified to recover the full bias matrix and thus to identify all causal effects.

\subsection{Modeling Spatial versus Temporal Correlation}
\label{sec:correlation_comparison}

So far, we have presented a model for residual spatial correlation in which we have replication in time. While this is typical in panel data settings, in some cases, where the residual spatial correlation is expected to be small (e.g. data with high resolution spatial sampling, as in the example in Section \ref{sec:application}), it may make more sense to apply our model to account for temporal correlation and treating spatial units as conditionally independent.  When modeling temporal correlation, all the results from the previous sections hold, but instead we have $\Cov(vec(\mathbf D) \mid \mathbf X, \mathbf S) = (BB^\top + \Lambda_D) \otimes I_N$ and $\Cov(vec(\mathbf Y) \mid \mathbf D, \mathbf X, \mathbf S) = (\Gamma \Gamma^\top + \Lambda_Y) \otimes I_N $, where $B$ and $\Gamma$ are $T \times M$ factor loading matrices.

In practice, we must choose whether to model residual spatial correlation and exploit replication in time, or instead model residual temporal correlation by leveraging replication over spatial units.  Somewhat surprisingly, we find that empirically, point estimates of causal effects are not particularly sensitive to whether dependence is modeled across space or across time. Finally, we note that our method can, in principle, be extended to models that account for low-rank residual correlations in both space and time. We leave this for future exploration, as covariance estimation and model selection are nontrivial in more general matrix normal models \citep[see e.g.][]{hoff2016limitations, efron2009set}.

\subsection{Estimation}
\label{sec:estimation}

In this work we primarily focus on  a fully Bayesian model for $(\mathbf D,\mathbf Y)\mid(\mathbf X,\mathbf S)$, which can be implemented using a probabilistic programming language such as Stan \citep{stan}. Prior distributions are placed on all model parameters, including the factor loading matrices $B$ and $\Gamma$. The likelihood is specified according to an outcome model based on the moments \eqref{eqn:outcome_mean}-\eqref{eqn:outcome_var} and the probabilistic PCA exposure model \eqref{eqn:exposure_factor} - \eqref{eqn:exposure_residual}. Inference via MCMC (or variational Bayes) leads to posterior distributions for all parameters, including the causal estimands of interest. By Theorem \ref{thm:identifiability_shared_confounding}, these effects are identifiable given sufficient assumptions on the factor structure and interference pattern (Assumptions~\ref{asm:ppca} - \ref{asm:rank}). Even when the causal effects are not fully identified, posterior distributions will reflect appropriate uncertainty within the partial identification region implied by Proposition \ref{prop:bias_bound}. This method is applied in the simulation studies (Section~\ref{sec:simulation}) and in the PM$_{2.5}$–birthweight analysis (Section~\ref{sec:application}).

When estimating linear causal effects, we can reduce bias by first orthogonalizing the exposure and the outcome with respect to observed covariates using double machine learning (DML; \citealp{chernozhukov2018double}) and then fitting the low-rank factor model to the residuals. Nonlinear or heterogeneous effects can be accommodated either through a precomputed basis expansion or within a fully probabilistic framework. In our application, we adopt a separable specification: the exposure is standardized and a centered cubic B-spline basis is constructed outside the sampler. The outcome mean then combines a linear term in the exposure with spline-based nonlinear components, preprocessed covariates, and the bias term, while the exposure and outcome covariances are jointly modeled through the latent factor structure. For continuous exposures, we summarize posterior draws by reporting the centered dose–response curve $g(d)-g(\bar d)$ with pointwise credible bands, the marginal effect \(g'(d)=\partial g(d)/\partial d\), the average causal derivative $\mathrm{ACD} = \E[g'(D)]$, and the average treatment effect of a shift, $\mathrm{ATE}(\Delta) = \E[g(D+\Delta) - g(D)]$. These quantities are computed from analytic derivatives and spline differences over the empirical exposure distribution. 

Finally, as an alternative to the fully Bayesian approach, we also propose a non-Bayesian three-step estimator that may be more suitable for handling high-dimensional data or when more complicated nonlinear functional forms are required. Details of the algorithm and rank selection are provided in Appendix~\ref{sec:appendix_b}. 

\section{Simulation Study}
\label{sec:simulation}

In this section, we assess the performance of our proposed factor confounding (FC) model in a number of simulated examples. We begin in Section~\ref{sec:sim_spatial_invariant} by comparing the estimation error across several benchmark methods.  
In Section \ref{sec:sim_misspec}, we examine the robustness of our method under model misspecification.

\subsection{Comparison of Estimators}
\label{sec:sim_spatial_invariant}

We first compare FC with alternative estimators under data generated from the proposed model \eqref{eqn:factor_y} and \eqref{eqn:exposure_factor}--\eqref{eqn:exposure_residual}, with $N=50$ locations, $T=100$ time points, and $M=3$ unmeasured confounders. For simplicity, we assume the causal effect is homogeneous and linear, $g_i(\mathbf{D}_t)=D_{it}$, with no interference, while exposures and outcomes may contain additional spatiotemporal components.

We compare seven estimators. The first three are baseline approaches that treat unmeasured confounding as a smooth function of location (and time) following \citet{gilbert2021causal}, and apply DML either (i) separately at each time point (``Single DML''), (ii) jointly across all locations (``Multi DML''), or (iii) on the entire stacked \(NT\) observations (``Stacked DML''). We then consider four additional estimators that incorporate factor modeling procedure to account for unmeasured confounding: (iv) a pure \(M\)-factor model fitted jointly to $(\mathbf D, \mathbf Y)$ (``FC(M)''); (v) a hybrid approach that first partials out spatiotemporal variation with DML before fitting the \(M\)-factor model (``FC(M)+DML''); (vi) the IFE estimator with $M$ latent factors (``IFE(M)''); and (vii) an IFE variant applied to DML residuals (``IFE(M)+DML''). 


\begin{figure}[htbp]	
    \centering
    \begin{subfigure}[t]{0.48\textwidth}
        \centering		\includegraphics[width=\textwidth]{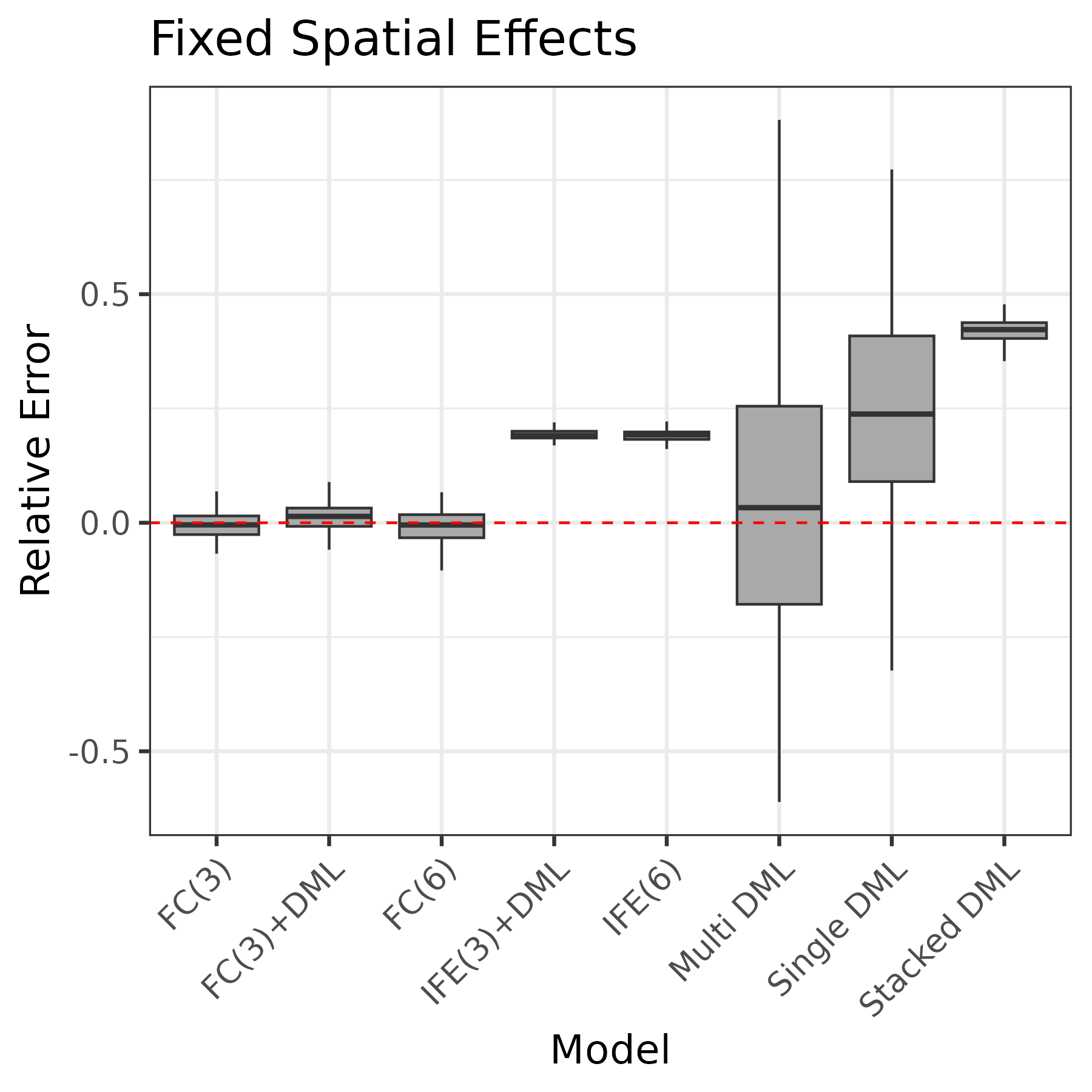}  \caption{\label{fig:linear_factor_error1}}
    \end{subfigure}
    \begin{subfigure}[t]{0.48\textwidth}
        \centering
        \includegraphics[width=\textwidth]{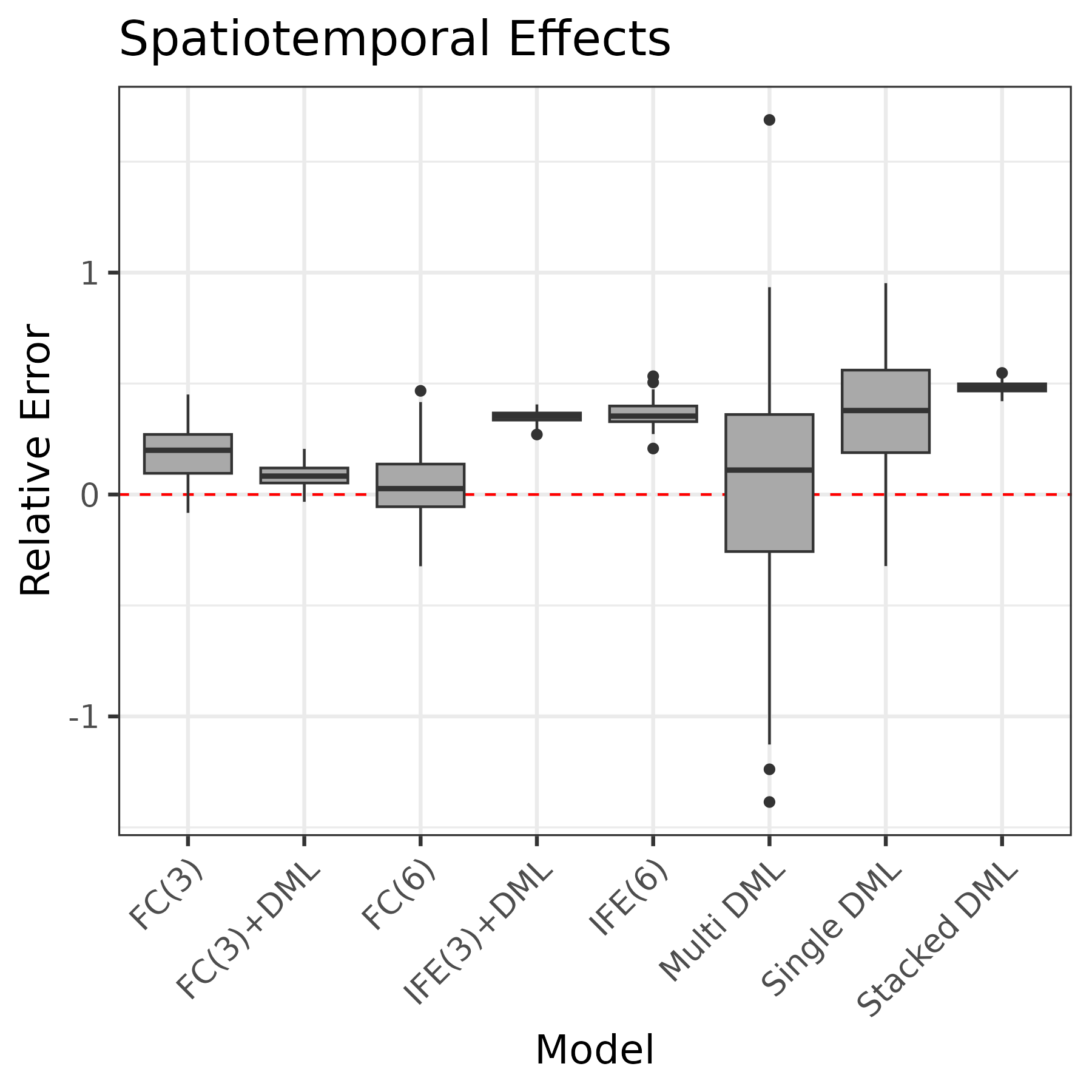}
    \caption{\label{fig:linear_factor_error2}}
    \end{subfigure}

    \caption{Relative error of estimated causal effects under (a) fixed spatial effects and (b) spatiotemporal effects.
    \label{fig:linear_factor_spatial_sim}}
\end{figure}

Figure~\ref{fig:linear_factor_spatial_sim} summarizes estimation error across $100$ simulated datasets. The naive DML variants are biased because smooth adjustment for observed location and time does not recover the residual low-rank confounding structure. The IFE estimator is suboptimal when unmeasured confounding dominates the idiosyncratic noise, or when the panel is short (see Figure~\ref{fig:sim_IFE_vs_FC} in Appendix~\ref{sec:appendix_ife_fc} for a more detailed comparison). In contrast, the hybrid factor confounding approach and the pure factor model with a sufficient number of latent factors yield the most accurate causal effect estimates. Sensitivity of FC to factor-rank selection is examined in Appendix~\ref{sec:sim_rank}, Figures~\ref{fig:sim_rank}--\ref{fig:sim_rank_sel}.

Appendix~\ref{sec:appendix_sim} provides additional simulations. Under spatial neighborhood interference, only FC recovers both the direct and spillover effects without appreciable bias. We further include an illustrative setting with heterogeneous nonlinear effects, in which FC recovers causal effects that vary across spatial units and depend nonlinearly on the exposures.

\subsection{Factor Confounding Under Model Misspecification}
\label{sec:sim_misspec}

We next investigate robustness of the FC estimator to violations of both the Gaussian exposure model in Assumption~\ref{asm:ppca} and the linear factor structure in Assumption~\ref{asm:linear_U}. To assess sensitivity to distributional misspecification while preserving the linear low-rank structure, we consider Student-$t$, Uniform, skewed normal, and discrete distributions for the latent confounders and residual errors. We also examine a setting with heteroskedastic residuals in both the exposure and outcome models. In addition, to assess sensitivity to structural misspecification, we introduce two departures from Assumption~\ref{asm:linear_U}: (i) nonlinear dependence of the exposures on the latent confounders and (ii) nonlinear confounding in the outcome model. All distributions are centered and rescaled to have unit variance so that comparisons isolate differences in distributional shape rather than scale.

Broadly, we find that when the outcome or exposure is nonlinear in the confounders, our linear method still works well as long as we inflate the ranks.  As such, we compare the performance of our approach under two different rank-selection strategies.  Our default approach is to selecgt the rank $\widehat M$ using a Bai--Ng-style information criterion applied to a reduced-form joint panel. As an alternative which is more robust to nonliearities, we also report an enriched, identification-compatible rank \(M_{\mathrm{enrich}}=\min\{2\widehat M,M_{\max}\}\), where \(M_{\max}\) is the largest rank satisfying the identification constraints. As shown in Figures~\ref{fig:sim_nongaussian} and \ref{fig:sim_nonlinear}, in these simulations FC(selected) and FC(enriched) are nearly unbiased across the non-Gaussian settings, while NUC and IFE are upwardly biased.

Under nonlinear outcome confounding, both FC variants remain approximaitely unbiased. Under nonlinear exposure confounding, especially the pure cubic transformation, FC(selected) has a small residual downward bias because the induced exposure signal is only approximately low-rank in a linear factor representation; FC(enriched) largely removes this bias with only a modest increase in dispersion.

\section{\texorpdfstring{The Causal Effect of PM$_{2.5}$ on Birth Weight}{The Causal Effect of PM2.5 on Birth Weight}}
\label{sec:application}

Birth weight is a crucial determinant of infant and long-term health, with low birth weight recognized as a significant risk factor for increased mortality and developmental complications. A recent umbrella review on air pollution and birth outcomes \citep{nyadanu2022prenatal} reinforces the growing evidence linking prenatal PM$_{2.5}$ exposure to reduced birth weight. However, estimated effects vary across studies due to differences in populations, exposure measurement, and analytical approaches. For instance, a meta-analysis by \citet{gong2022maternal} estimated that a $10$~$\mu g/m^3$ increase in PM$_{2.5}$ exposure during pregnancy leads to an average birth weight reduction of $17$~$g$ (95\% CI: [$-20$, $-13$]), with high between‑study heterogeneity (range from $-79$~$g$ to $25$~$g$). Other meta-analyses \citep{ghosh2021ambient, uwak2021application} reported greater reductions of $22$~$g$ (95\% CI: [$-32$, $-12$]) and $28$~$g$ (95\% CI: [$-48$, $-7$]) for the same exposure increment. 
\citet{gilbert2021causal} 
estimate a $-16$ $g$ (95\% CI: [$-28$, $-3$]) reduction using a nonlinear heterogeneous DML estimator assuming unmeasured confounding is a measurable function of the spatial location. These modest discrepancies highlight the challenges posed by confounding and measurement variability in existing research, underscoring the need for robust causal inference methods to more accurately quantify the impact of prenatal PM$_{2.5}$ exposure on birth outcomes.

In this study, we analyze ZIP code-level birth weight data spanning from $2018$ to $2024$, obtained from the California Vital Data (Cal-ViDa) Query Tool \citep{calvida_birth}. The public data report counts in birth weight categories rather than individual birth records, and small cells are suppressed to protect confidentiality. We therefore impute these values using a quasi-Poisson model, followed by normalization to match the reported total births within each ZIP–year stratum. We then approximate the birth weight distribution within each stratum using a normal density and estimate the mean by maximizing the empirical likelihood. Finally, ZIP codes with fewer than $11$ total births in any year are excluded from the analysis, resulting in a dataset of $1,209$ ZIP codes observed over seven years. The exposure variable, annual ZIP-level PM$_{2.5}$ concentration, is calculated from high-resolution ($0.01^\circ \times 0.01^\circ$) estimates developed by \citet{shen2024enhancing}, which integrate information from satellite, simulation, and ground-based monitoring data.


We adjust for $26$ potential confounders, including both environmental and socioeconomic variables measured at the ZIP code level. Parental demographic covariates, constructed from stratified birth count data provided by Cal-ViDa, capture the proportions of births by maternal age, race/ethnicity, education level, nativity, and trimester of prenatal care initiation. The annual median household income is obtained from \citet{manson2024ipums}. Geographic coordinates and calendar year are retained in all specifications so that broad spatial and temporal structure is modeled before the residual factor adjustment is applied. These covariates capture important demographic composition and healthcare access patterns across ZIP codes, but may not fully account for all sources of confounding in the relationship between PM$_{2.5}$ and birth weight. Residual confounding may persist due to additional unmeasured factors such as overall maternal health conditions, health-related behaviors, and parity. Regional factors including industrial activity, healthcare accessibility, and environmental conditions (e.g., green space, temperature, humidity), some of which remain stable over the study period, may also contribute to unmeasured confounding.

In this analysis, we start by modeling the causal effect using a variant of generalized additive models that accommodates nonlinear PM$_{2.5}$ effects while adjusting for factor confounding (see Section~\ref{sec:estimation}). In contrast to the simulations in Section \ref{sec:simulation}, given the high density of spatial locations and small number of yearly observations, we treat ZIP codes as independent observations and apply the factor model to capture residual temporal correlation instead.  For the factor model, we fix the number of latent factors at $M=3$, reflecting two considerations: (i) identification under the factor confounding assumption with seven years of data requires at most rank~3; and (ii) in practice, modestly overestimating the number of latent factors rarely degrades performance. 

\begin{figure}[h]
  \centering
  \includegraphics[width=\textwidth]{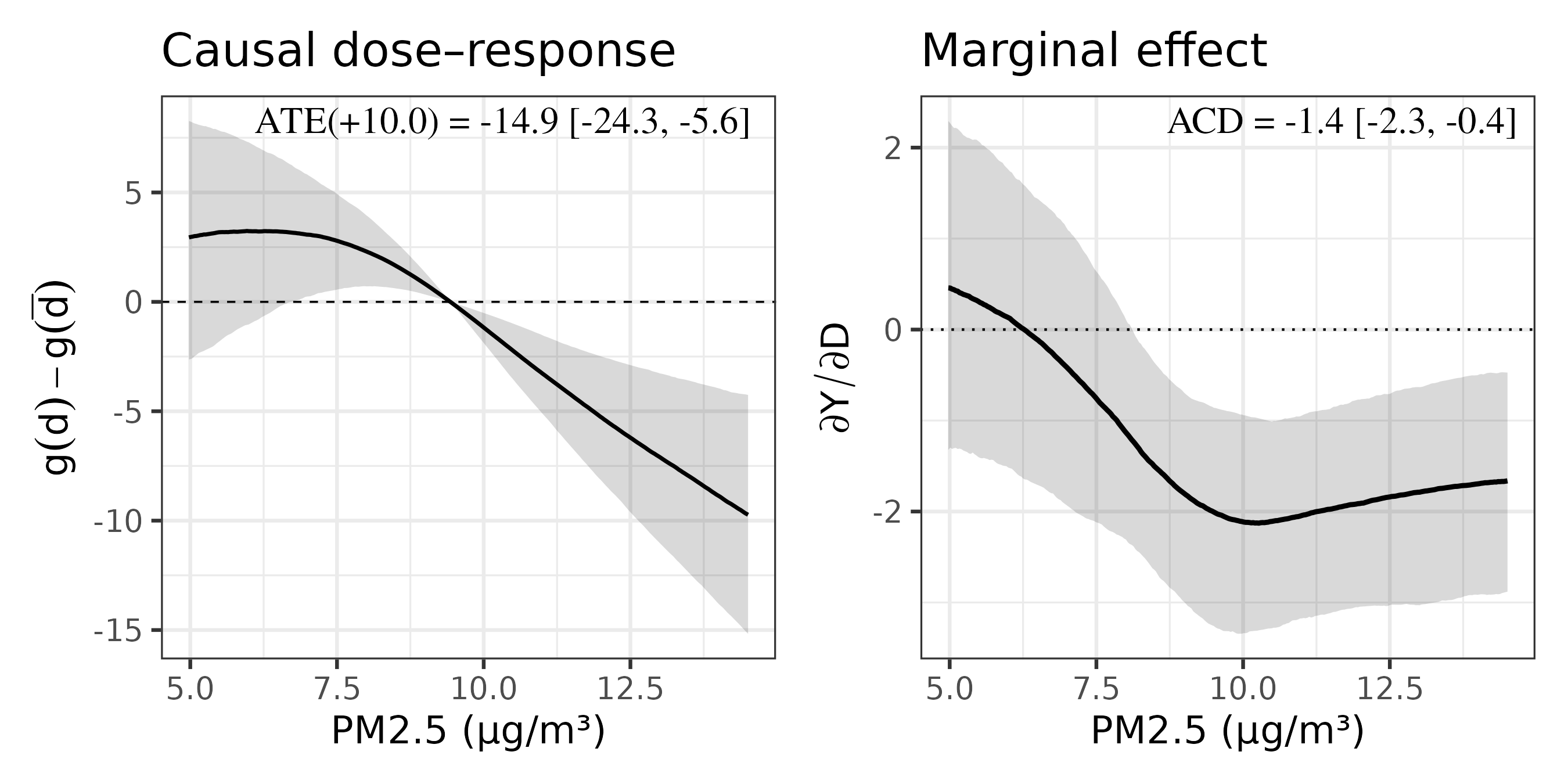}
    \caption{Estimated causal dose-response curve and marginal causal effect under factor confounding.  Left) The centered causal dose-response function.  On average, a $10$ $\mu g / m^3$ increase in PM$_{2.5}$ exposure corresponds to a $14.9$ $g$ decrease in birth weight. Right) The marginal effect of PM$_{2.5}$ on birth weight, as a function of exposure. The average causal derivative is about $-1.4$.  Both plots indicate little evidence of an effect at the lowest exposure levels, but a reduction of about 2~$g$ in birth weight per unit increase in PM$_{2.5}$ above the median exposure level.}
  \label{fig:bw_nonlinear}
\end{figure}

Figure \ref{fig:bw_nonlinear} shows the inferred dose-response function (left) and marginal effect (right) of \pmtf on birth weight after the factor confounding adjustment.  We estimate the average treatment effect for a $10$ $\mu g /m^3$ change in \pmtf as $\E_D[g(D + 10) - g(D)] = -14.9$ $g$ with a 95\% credible interval of $(-24.3, -5.6)$.  The average causal derivative is $\E_D[g'(D)]=-1.4$ with a 95\% credible interval of $(-2.3, -0.4)$.  Both of these estimates are broadly in agreement with previously published estimates.  In general, we find that below the population median \pmtf level of $9.3$ $\mu g/m^3$, there is little evidence that increases in the exposure cause changes in birth weight. However, above the median, a unit increase in \pmtf causes an average reduction of approximately 2 $g$.

\begin{figure}[ht]
  \centering
    \includegraphics[width=\textwidth]{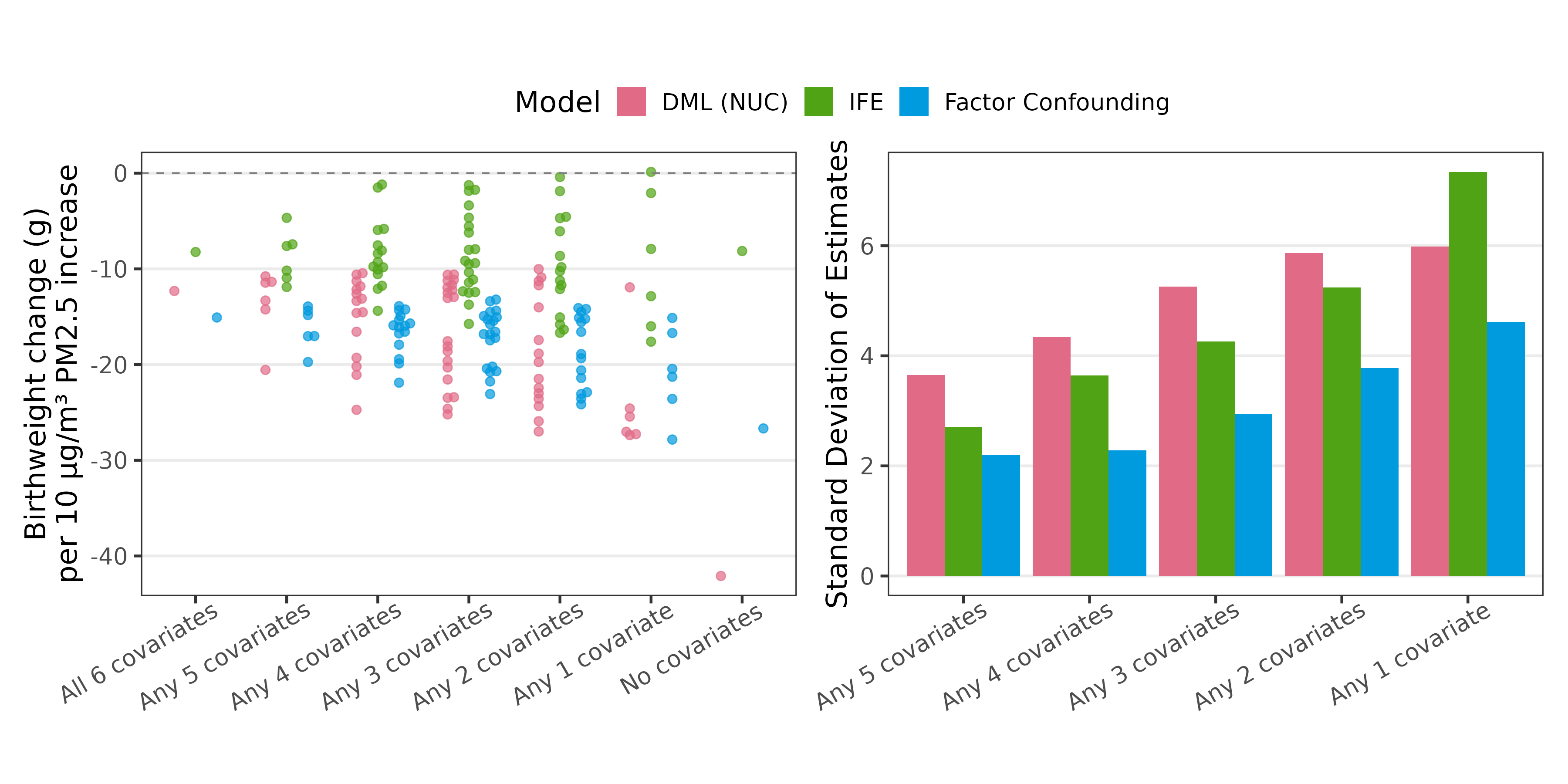}
    \caption{Left) Comparison of effect estimates when controlling for all $\binom{6}{k}$ subsets of confounders, using standard DML, DML combined with the IFE model, and DML combined with factor confounding. Right) Standard deviations of estimates with different number of covariates excluded.  In all cases, factor confounding estimator exhibits the least variability under the exclusion of potential confounders, suggesting greater robustness to omitted variables than either IFE or the standard spatial confounding approaches.}
  \label{fig:bw_combined}
\end{figure}

To demonstrate the robustness of the factor confounding approach on this dataset, we next compare the causal effect estimates when explicitly leaving out potentially important measured confounders (Figure~\ref{fig:bw_combined}). By leaving out different subsets of covariates, we can assess how sensitive each estimator is to the presence of unmeasured confounding. We always adjust for spatial location and year, while varying the exclusion of the six remaining covariate groups. We compare three different estimators: DML under the NUC assumption, the IFE model, and our factor confounding method. For comparability with IFE, the factor confounding adjustment is applied within a partially linear specification in which the outcome is modeled as linear in the exposure. In the left panel of Figure \ref{fig:bw_combined}, we plot the causal effect estimates, grouped by the number of included covariate groups and colored by the estimator. Unsurprisingly, variation in estimates is greatest when the fewest covariates are included. In contrast, each estimator yields comparable estimates when all covariates are included, presumably because we reduce the residual correlation by conditioning on more potential confounders.  The right panel of Figure \ref{fig:bw_combined} shows the standard deviation of estimates by estimator and the number of included covariates, as another way to visualize the sensitivity of each estimator to potential unmeasured confounding. DML (NUC), which does not incorporate any adjustment for unmeasured confounding exhibits substantial sensitivity to the adjustment set, exceeding the variability of both factor confounding and IFE estimates in almost all settings. IFE exhibits the second highest variability in estimates in nearly all settings, except when only one covariate group is included, where it shows the greatest variability.  In contrast, the factor confounding model consistently yields the most stable estimates, showing the least sensitivity to the choice of included covariates relative to IFE and DML (NUC).

Further, although the true causal effect is unknown, estimates from the factor confounding approach align more closely with previously published results. In Appendix Figure~\ref{fig:sensitivity_rmse}, we plot the RMSE of the estimates against the hypothetical causal effect across all $\binom{6}{k}$ covariate subsets, colored by model and faceted by $k$.  Factor confounding consistently achieves the lowest RMSE for nearly all values of $k$ when the true causal effect is about $-15$ $g$ or smaller, whereas  IFE tends to perform better when the effect exceeds $-15$ $g$.  Prior meta-analyses and causal studies suggest effects in the range of $-16$ $g$ to $-28$ $g$ \citep{ghosh2021ambient, gilbert2021causal, uwak2021application, gong2022maternal}, supporting the plausibility of the factor confounding estimates.

Finally, we conduct an additional sensitivity analysis to assess whether the estimated birth weight effects are influenced by spillovers from nearby areas or prior periods. Specifically, we augment the analysis with two additional exposure measures: the average PM$_{2.5}$ level among the five nearest neighboring ZIP codes and prior-year PM$_{2.5}$ exposure. The estimated spillover effects are small and not statistically significant across all methods (Appendix Figure~\ref{fig:bw_interference}). However, for the direct effect of local exposure, only FC remains stable and statistically significant, whereas DML (NUC) and IFE fail to recover the signal.


\section{Discussion}


In this work, we introduced a factor confounding framework for spatiotemporal causal inference with continuous exposures, shared unmeasured confounders, and possible interference. Building on ideas from causal panel data and multivariate causal inference, the proposed model uses low-rank factor structure to characterize bias from shared confounders. When exposures and outcomes are linear in the unmeasured confounders, this structure partially identifies causal effects while allowing treatment effect heterogeneity, nonlinear exposure effects, and interference. We further show that additional (typically weak) restrictions on neighborhood interference can yield point identification.

There are several extensions worth considering. Our theoretical results rest on several assumptions such as linearity in the unmeasured confounders, Gaussianity of residuals, and the absence of treatment–confounder feedback \citep{hernan2020causal}.  We show in simulations that the results are not sensitive to the Gaussian and linear confounding assumptions, but leave any formal extensions involving non-Gaussian or nonlinear relationships for future work. Methods which incorporate our approach into estimation procedures in the presence of treatment-confounder feedback are also warranted. Indeed, in the birth weight application (Section \ref{sec:application}), the observed covariates do not perfectly satisfy the factor confounding assumption, which explains why estimates vary when subsets of confounders are systematically held out.  To address such concerns, our approach could be combined with additional sensitivity analyses \citep[e.g. in the spirit of][]{cinelli2020making}. Developing sensitivity analyses tailored to spatiotemporal causal settings remains an important direction for future research.


\section*{Code Availability}

The code used to implement the proposed methods and reproduce the simulation studies and birth weight application is available at
\url{https://github.com/jw1212/spatiotemporal-factor-causal}.

\bibliography{references,references_spatial}

\clearpage

\appendix
\section{Proofs}
\subsection*{Proof of Theorem~\ref{thm:identifiability_shared_confounding}}





\begin{proof}
\textbf{Step 1 (Characterization of bias).}
Without loss of generality, we assume $\E[\mathbf D_t]=0$. Joint Gaussianity of $(\mathbf U_t,\mathbf D_t)$ implies
\[
\E[\mathbf U_t\mid\mathbf D_t]
   =B^{\top}\Sigma_{D}^{-1}\mathbf D_t,
\quad
\Sigma_{U\mid D}:=\Cov(\mathbf U_t\mid\mathbf{D}_t) = I_M-B^{\top}\Sigma_{D}^{-1}B \succ0.
\]
Substituting these into Equation \eqref{eqn:factor_y} gives the conditional mean 
\[
\E[\mathbf Y_t\mid\mathbf D_t] =g(\mathbf D_t)+C\mathbf D_t,
\]
where $C = \Gamma\Sigma_{U\mid D}^{-1/2}B^{\top}\Sigma_{D}^{-1}$ controls the confounding bias.

Since $B$ and $\Gamma$ are identified up to independent rotations under Assumption~\ref{asm:factor}, for any rotations $\tilde{B}=B\Theta_1,\tilde{\Gamma}=\Gamma\Theta_2, \Theta_1, \Theta_2 \in \mathcal{O}_M$, the bias can be written as 
\[
\tilde{C} = \tilde{\Gamma}\tilde{\Sigma}_{U\mid D}^{-1/2}\tilde{B}^{\top}\Sigma_D^{-1} = \Gamma\Theta\Sigma_{U\mid D}^{-1/2}B^{\top}\Sigma_D^{-1},
\] 
where $\Theta := \Theta_2\Theta_1^\top$ is still an orthogonal matrix. Thus we can fix $B$ and $\Gamma$, then $\Theta$ remains the only degree of freedom in bias. Note that $C$ is invariant if $\Theta_1=\Theta_2$, so $B$ and $\Gamma$ are at most identified up to a common rotation from the right.

\textbf{Step 2 (Off-neighborhood identification of $C$).}
Under partial interference, $g_i(\cdot)$ depends on the exposure vector only through its neighborhood $\mathcal N_i$, so $C_{ij}$ is identified from 
\[
\partial_{d_j}\E[Y_{it}\mid\mathbf D_t=\mathbf d] =C_{ij}, \quad \forall i \in \mathcal{I},\, j\notin \mathcal N_i.
\]
To resolve the rotation indeterminancy and identify the full matrix $C$, it suffices to show that the mapping from $\Theta$ to these identified entries of $C$ is injective. 

\textbf{Step 3 (Injectivity w.r.t.\ rotations).}
For any $\Theta\in\mathcal O_M$, write
\[
C(\Theta)=\Gamma\Theta R, \quad C_{ij} = \gamma_i^\top\Theta r_j. 
\]
Consider two rotation matrices $\Theta$ and $\Theta'$. The off-neighborhood identification of $C$ gives $C_{ij}(\Theta)=C_{ij}(\Theta')$, $\forall j\notin \mathcal N_i$, hence
\[
[\Gamma(\Theta-\Theta')R]_{ij} = \gamma_i^\top(\Theta-\Theta')r_j = \langle\Theta-\Theta', \gamma_i r_j^{\top}\rangle_F=0,
\]
where $\langle\cdot,\cdot\rangle_F$ is the Frobenius inner product. Define the set of outer products 
\[S := \bigl\{\gamma_i r_j^{\top}: j\notin \mathcal N_i\bigr\} \subset\mathbb R^{M\times M}.
\]
By Assumption~\ref{asm:rank}, there exists $M$ linearly independent rows of $\Gamma$ and sufficient rank in off-neighborhood columns of $R$ to generate a basis of $\mathbb R^{M \times M}$, which implies $\operatorname{span}(S)=\mathbb R^{M\times M}$. Therefore, it follows that $\Theta=\Theta'$ and the mapping $\Theta \mapsto C(\Theta)$ is injective.



\textbf{Step 4 (Identification of $g$).}
For each $i$, taking the conditional expectation with respect to $D_{it}$ yields
\[
\E\left[Y_{it}-(C\mathbf D_t)_i\mid D_{it}\right]=g_i(D_{it}),
\]
which is identified once $C$ is known. Non-parametric estimation of $g_i$ requires the usual smoothness and support conditions, which are assumed satisfied. All claims are proved.

\end{proof}

\clearpage

\section{Estimation Details}
\label{sec:appendix_b}

Algorithm \ref{alg:factor_no_spillover} outlines the three-step estimator for causal inference under factor confounding. For uncertainty quantification, we augment this procedure with an additional bootstrap step.

\begin{algorithm}[h]
\caption{Three-step estimator under factor confounding}
\label{alg:factor_no_spillover}
\textbf{Input:} i.i.d.\ observations $\{(\mathbf Y_t,\mathbf D_t, \mathbf X_t)\}_{t=1}^T$, number of latent factors $M$, bootstrap replicates $B$, number of random initializations $N_{\text{init}}$. \\
\textbf{Output:} bias matrix $\widehat C$, dose-response functions $\{\widehat g_i\}_{i=1}^{N}$, causal estimand(s) $\widehat\beta$, and corresponding confidence regions.

\textbf{Step I: Off–neighborhood entries of $C$.} \\ 
\quad For each $i=1,\dots,N$, fit the partially linear model
\[
Y_{it} \;=\; m_i(D_{it},\mathbf X_t) +\sum_{j\notin \mathcal N_i}\alpha_{ij}D_{jt} + \xi_{it},
\]
where $m_i(\cdot)$ is learned nonparametrically. Set $\widehat C_{ij}\leftarrow\widehat\alpha_{ij}$ for $j\notin \mathcal N_i$ and denote the resulting off–neighborhood matrix by $\widehat C_{\mathrm{off}}$.

\textbf{Step II: Neighborhood entries of $C$ via factor analysis.}
\begin{enumerate}[itemsep=0.25em, topsep=0.25em]
    \item Regress each exposure $D_{it}$ on $\mathbf X_t$, collect residuals $\widetilde{\mathbf D}_{t}$, and fit an $M$‑factor model to obtain loadings $\widehat B$.
    \item Regress each outcome $Y_{it}$ on $(D_{it},\mathbf X_t)$, collect residuals $\widetilde{\mathbf Y}_t$, and fit an $M$‑factor model to obtain loadings $\widehat\Gamma$.
    \item Compute $\widehat R=\widehat\Sigma_{V\mid D,X}^{-1/2}\, \widehat B^{\top}\widehat\Sigma_D^{-1}$ and solve the generalized orthogonal Procrustes problem
    \[
    \widehat\Theta =\arg\min_{\Theta\in\mathcal O_M} \bigl\| (\widehat\Gamma\Theta\widehat R)_{\mathrm{off}} -\widehat C_{\mathrm{off}} \bigr\|_F^{2}.
    \]
    Repeat the optimization $N_{\text{init}}$ times with independent initializations $\widehat{\Theta}_0$ and retain the optimal solution. Construct $\widehat C=\widehat\Gamma\widehat\Theta\widehat R$ and overwrite its off-neighborhood with $\hat{C}_{\mathrm{off}}$.
\end{enumerate}

\textbf{Step III: Debiasing and estimation of $\{g\}$.}

Obtain the debiased outcomes $Y_{it}^{\dagger} =Y_{it}-\sum_{j=1}^{N}\widehat C_{ij}\widetilde{D}_{jt}$. For each $i=1,\dots,N$, fit a model $\widehat g_i(\cdot)$ of $Y_{it}^{\dagger}$ on $(D_{it},\mathbf X_t)$ and compute $\widehat\beta$.

\textbf{Step IV: Bootstrap inference.} \\
Draw $B$ bootstrap samples. For each sample $b=1,\dots,B$, repeat \textbf{Steps I–III} to obtain $\{\widehat C^{*(b)},\widehat g^{*(b)}, \beta^{*(b)} \}_{b=1}^B$ and form confidence regions.

\end{algorithm}

The number of latent factors is a tuning parameter in practice and can be selected using one or more of the following tools:
\begin{itemize}
  \item \textbf{Exploratory diagnostics.} Scree plots, eigenvalue-ratio tests, or parallel analysis can help identify the point beyond which additional factors explain little incremental variation.
  \item \textbf{Information criteria.} Classical criteria such as AIC or BIC, as well as the factor-selection criteria of \citet{bai2002determining}, can be used to balance model fit and complexity.
  \item \textbf{Stability checks.} One may examine how estimated factor loadings and causal effect estimates vary with the fitted rank, and retain the smallest rank at which these quantities stabilize.
  \item \textbf{Bayesian comparison.} Candidate Bayesian factor models satisfying the factor-confounding structure may be compared using Pareto-smoothed importance sampling leave-one-out predictive loss (PSIS--LOO; \citealp{vehtari2017practical}).
\end{itemize}

In practice, we recommend combining these diagnostics with substantive judgment and sensitivity analysis. Our simulations suggest that moderate over-specification of the factor rank typically has little effect on bias, provided the identification conditions remain satisfied, although it may increase variance somewhat. By contrast, under-specification can induce appreciable bias when the fitted rank is too small to capture the latent confounding structure. In the misspecification simulations below, we therefore report both a parsimonious data-driven rank, denoted \(M=\widehat M\), and an enriched rank \(M_{\mathrm{enrich}}=\min\{2\widehat M,M_{\max}\}\), where \(M_{\max}\) is the largest rank compatible with the relevant identification constraints. The enriched fit is intended as a sensitivity analysis for approximate low-rank structure: it allows additional latent directions when the selected rank is too parsimonious, while avoiding ranks outside the identified region.



\clearpage

\section{Additional Simulation Results}
\label{sec:appendix_sim}

\subsection*{Metrics for Section~\ref{sec:sim_spatial_invariant}}




\begin{table}[htbp]
\centering
\caption{Bias, standard deviation, and mean squared error for the estimators considered in the homogeneous linear effect setting with factor confounding.} 
\label{table:sim_linear}
\begin{tabular}{llrrr}
\toprule
\textbf{Setting} & \textbf{Method} & \textbf{Bias} & \textbf{SD} & \textbf{MSE} \\
\midrule
Fixed spatial effects 
& Single DML & 0.240 & 0.209 & 0.100 \\ 
& Multi DML & 0.011 & 0.386 & 0.146 \\
& Stacked DML & 0.415 & 0.026 & 0.173 \\
& \textbf{FC(3)} & \textbf{-0.005} & 0.030 & \textbf{0.0009} \\
& FC(6) & -0.012 & 0.038 & 0.0016 \\
& FC(3)+DML & 0.012 & 0.031 & 0.0011 \\
& IFE(6) & 0.193 & 0.014 & 0.0374 \\
& IFE(3)+DML & 0.194 & 0.013 & 0.0378 \\
\addlinespace
Spatiotemporal effects
& Single DML & 0.336 & 0.290 & 0.195 \\ 
& Multi DML & 0.053 & 0.515 & 0.260 \\
& Stacked DML & 0.482 & 0.027 & 0.233 \\
& FC(3) & 0.174 & 0.131 & 0.047 \\
& \textbf{FC(6)} & \textbf{0.032} & 0.092 & 0.010 \\
& \textbf{FC(3)+DML} & 0.077 & 0.048 & \textbf{0.008} \\
& IFE(6) & 0.362 & 0.056 & 0.134 \\
& IFE(3)+DML & 0.345 & 0.023 & 0.120 \\
\bottomrule
\end{tabular}
\end{table}

\subsection*{Comparison by Confounding Strength and Panel Length}
\label{sec:appendix_ife_fc}

For further comparison between DML (NUC), IFE and FC, we simulate panel data with latent factor confounding. Exposure and outcome are both driven by a small number of unobserved factors, plus idiosyncratic noise, with a constant causal effect $\beta=1$. To control confounding severity, we draw the factor loadings that link factors to exposure and outcome from a bivariate normal with correlation $\rho$, so larger $\rho$ implies stronger alignment between the two loading vectors and hence more severe latent confounding. We fix the number of units at $N=50$ and the factor dimension at $M=3$, and vary the panel length $T\in\{10,50,200\}$ to probe different $N/T$ regimes. Confounding strength varies over $\rho\in\{0,0.5,0.9\}$. Each $(\rho,T)$ configuration is replicated $50$ times. All methods are given the correct factor rank in these simulations. 

Figure~\ref{fig:sim_IFE_vs_FC} reports the sampling distributions of \(\widehat{\beta}\), faceted by confounding strength \(\rho\) and panel length \(T\). With weak confounding ($\rho=0$), all estimators concentrate near the truth, and dispersion shrinks as $T$ increases. At moderate confounding ($\rho=0.5$), the DML (NUC) baseline shows clear bias; IFE improves with longer panels as factor recovery stabilizes; and the joint FC model remains close to $\beta$ with smaller spread. Under severe confounding ($\rho=0.9$), DML (NUC) is substantially biased and can dominate the axis scale; IFE moves toward $\beta$ as $T$ grows but remains sensitive when $T$ is small; FC is the most robust overall, with medians near $\beta$ across $T$ and reduced dispersion. These patterns align with the identification logic: when exposure and outcome share strong low-rank structure, explicit factor adjustment is essential, and pooling information from both sides improves recovery of the confounder space, particularly in shorter panels.

\begin{figure}[htbp]
\centering
\includegraphics[width=0.8\linewidth]{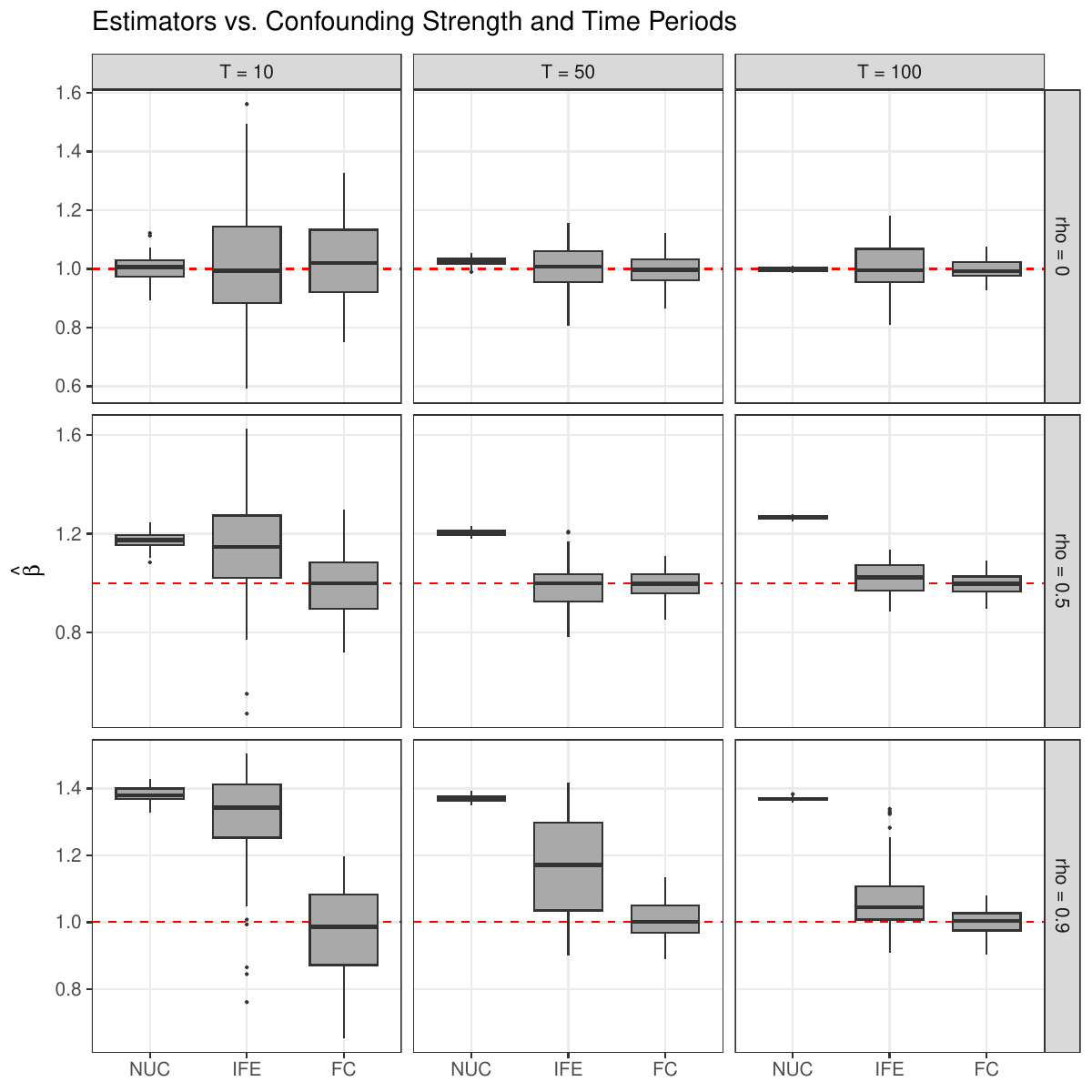}
\caption{Sampling distributions of $\widehat\beta$ across replications, faceted by confounding strength $\rho$ (rows) and panel length $T$ (columns). The dashed line marks the true effect $\beta=1$.}
\label{fig:sim_IFE_vs_FC}
\end{figure}


\subsection*{Sensitivity to Factor-Rank Specification}
\label{sec:sim_rank}

To examine sensitivity to factor-rank specification within the correctly specified factor confounding model, we simulate panel data under spatially invariant confounding with true factor rank \(M \in \{3,6,10\}\), fixing \(N=30\) and \(T=100\). We compare FC fitted with fixed ranks $3$ and $6$, together with a data-driven fit FC(selected), where the factor rank is chosen from a prespecified candidate set using a Bai--Ng-style information criterion \citep{bai2002determining}, applied to a reduced-form joint panel. Specifically, we first form \(W=[D,\;Y-\tilde\beta D]\), where \(\tilde\beta\) is the pooled least-squares slope from regressing \(Y\) on \(D\), and then select the rank by minimizing a penalized residual variance criterion over the candidate set. Figure~\ref{fig:sim_rank} reports the sampling distributions of \(\widehat\beta - \beta\), and Figure~\ref{fig:sim_rank_sel} summarizes the empirical distribution of the selected ranks. This simulation evaluates the behavior of the selector itself, while the subsequent misspecification simulations also include the enriched rank \(M_{\mathrm{enrich}}\) to examine whether a modest identification-compatible expansion improves robustness when the observed data is only approximately linear low-rank.

\begin{figure}[htbp]
\centering
\includegraphics[width=0.8\linewidth]{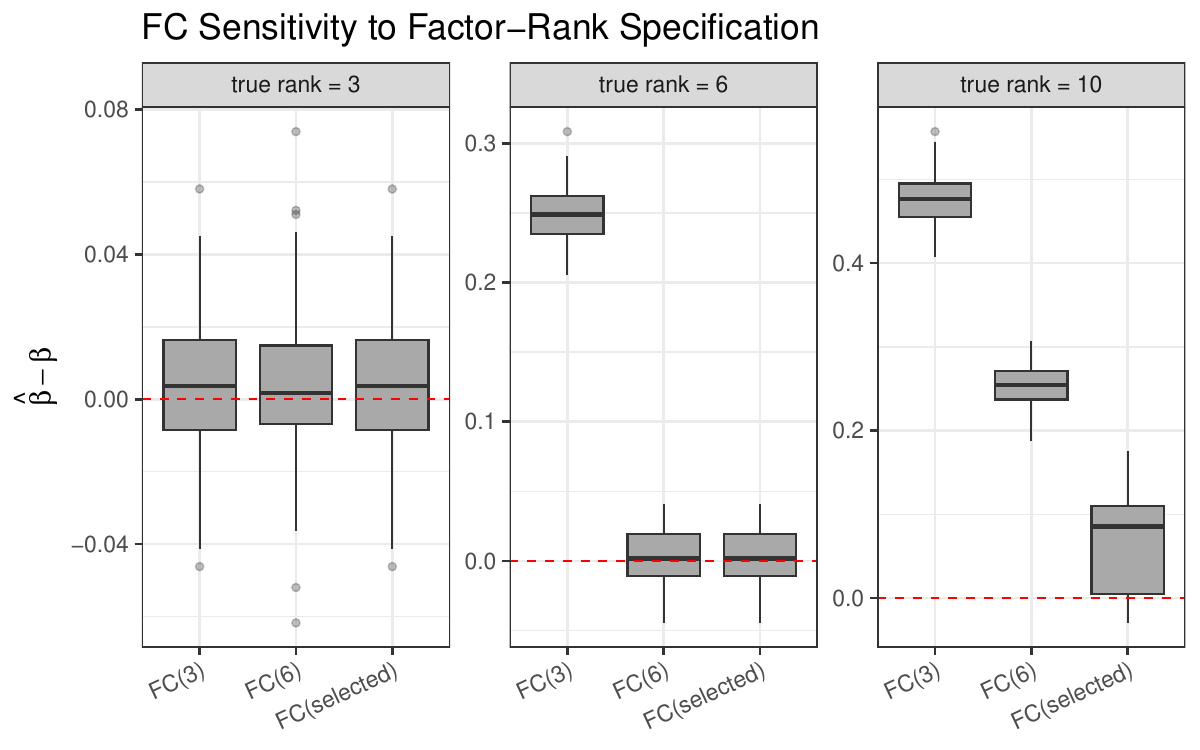}
\caption{Sampling distributions of estimation error under factor confounding, with true factor rank \(M \in \{3,6,10\}\). We compare FC fitted with fixed ranks 3 and 6 and a data-driven fit FC(selected).}
\label{fig:sim_rank}
\end{figure}

\begin{figure}[htbp]
\centering
\includegraphics[width=0.8\linewidth]{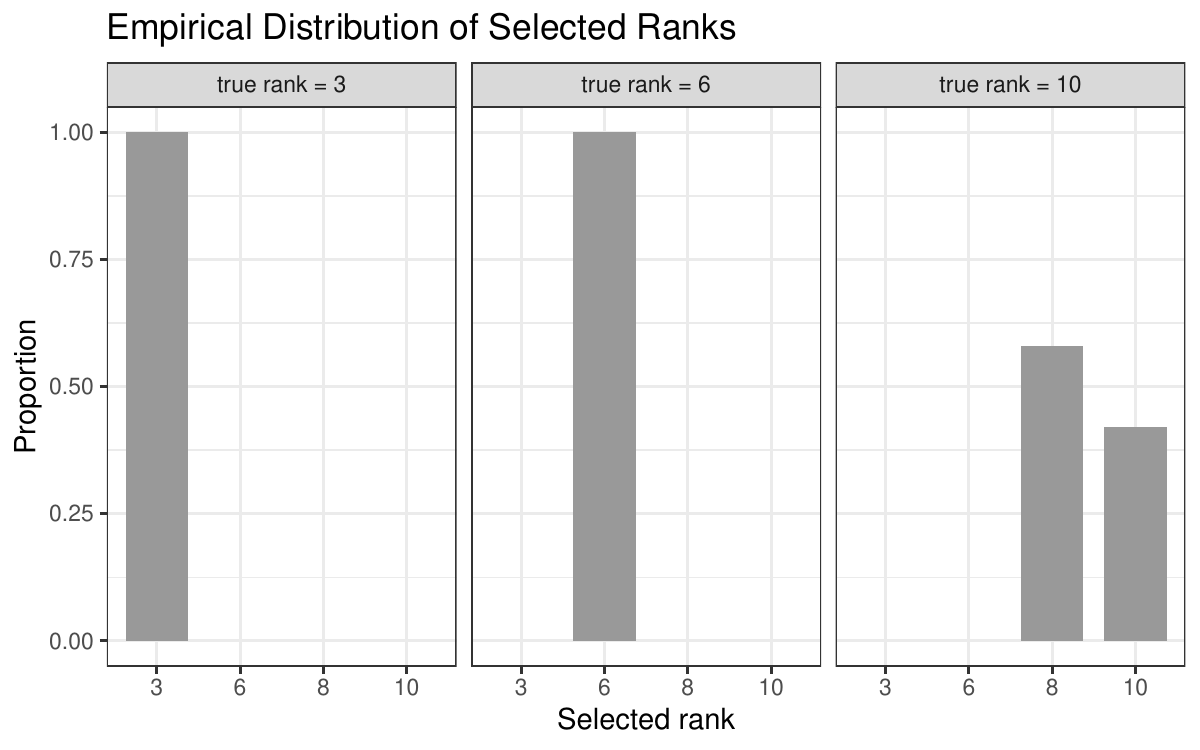}
\caption{Empirical distribution of the selected factor rank for FC(selected) across simulation replications, shown by true factor rank \(M \in \{3,6,10\}\).}
\label{fig:sim_rank_sel}
\end{figure}

When the true rank is 3, all FC variants perform similarly, indicating that moderate over-specification causes little loss in this setting. When the true rank is 6, FC(6) and FC(selected) remain close to unbiased, whereas FC(3) exhibits clear upward bias due to under-specification of the latent confounder dimension. When the true rank increases to 10, both FC(3) and FC(6) become substantially biased. FC(selected) partially adapts by choosing larger ranks more often and thereby reduces bias, but it does not eliminate it entirely. Overall, these results show that FC is reasonably robust to moderate over-specification of the factor rank, whereas under-specification can induce substantial bias once the true confounding dimension is high. They also suggest that data-driven rank selection can mitigate, though not fully remove, this sensitivity in finite samples.


\subsection*{Robustness to Model Misspecification}

We examine the robustness of the estimators to both distributional and structural departures from the baseline factor confounding model. In these simulations, data are generated with \(N=30\) spatial locations, \(T=100\) time periods, and \(M=3\) latent factors. Rather than fixing the fitted rank, we compare NUC, IFE, FC(selected), and FC(enriched). FC(selected) uses the Bai--Ng-style criterion described above, while FC(enriched) fits \(M_{\mathrm{enrich}}=\min\{2\widehat M,M_{\max}\}\).

To assess distributional departures from the Gaussian exposure model (Assumption~5) while preserving the linear factor structure, we vary the distributions of the latent factors and residual errors. Specifically, we consider Gaussian, Student-\(t_4\), Uniform, skewed-normal, heteroskedastic, and discrete latent factor settings. In the Student-\(t\), Uniform, and skewed-normal settings, both the latent factors and the residual errors are drawn from the corresponding distribution and then centered and rescaled to achieve the target variance. In the heteroskedastic setting, the latent factors remain Gaussian, while the exposure and outcome residuals are generated from scaled Student-\(t_4\) distributions with unit-specific variance multipliers. We also consider a discrete latent factor design in which \(\mathbf U_t\) take values in a finite set of \(K=5\) latent states in \(\mathbb{R}^M\) with unequal state probabilities, while the residual errors remained Gaussian; the resulting latent factors are then centered and standardized componentwise.

\begin{figure}[htbp]
  \centering
  \includegraphics[width=0.8\linewidth]{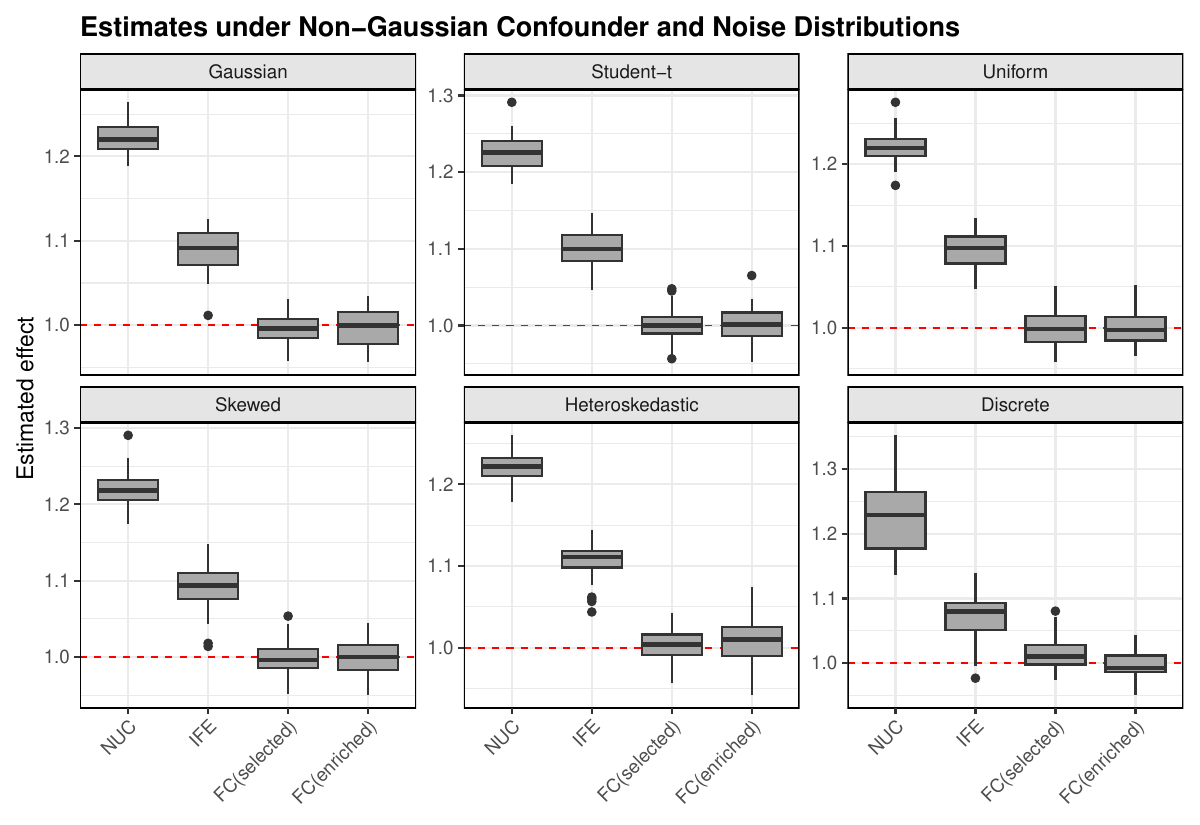}
  \caption{Sampling distributions of the estimated causal effect under non-Gaussian latent factor and residual distributions. The FC variants use the selected and enriched data-driven ranks, respectively. The dashed red line denotes the true causal effect.}
  \label{fig:sim_nongaussian}
\end{figure}

To study structural departures from the linear factor model (Assumption~4), we keep the latent factors and residual errors Gaussian but replace the linear factor index in either the exposure model or the outcome confounding term by a nonlinear transformation. In the nonlinear exposure settings, exposures are generated as
\[
\mathbf D_t = h(B\mathbf U_t) + \bm\xi_t,
\]
where \(h(\cdot)\) is applied componentwise to the linear exposure index and then centered and rescaled columnwise. We consider three choices of \(h\): a linear-plus-cubic transformation, a pure cubic transformation, and a \(\tanh\) transformation. Outcomes then remain linear in the latent confounders. In the nonlinear outcome settings, exposures remain linear in the latent factors, while the latent confounding term in the outcome equation is replaced by
\[
\mathbf Y_t(\mathbf d_t)
= \alpha\mathbf d_t + g(\Gamma \mathbf U_t) + \bm\varepsilon_t,
\]
where \(g(\cdot)\) is again applied componentwise and then centered and rescaled columnwise. Here we consider linear-plus-cubic, \(\sin\), and \(\tanh\) transformations.

\begin{figure}[htbp]
  \centering
  \includegraphics[width=0.8\linewidth]{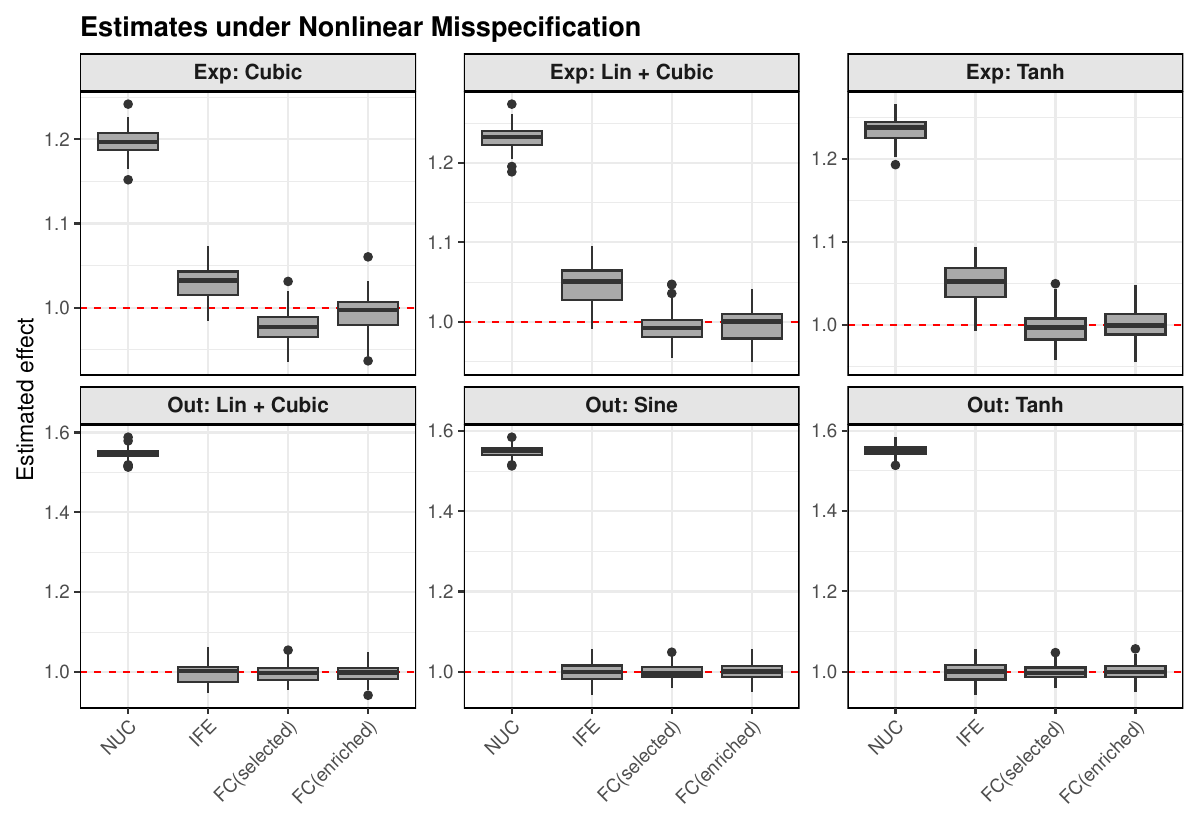}
  \caption{Sampling distributions of the estimated causal effect under nonlinear exposure and outcome models. The FC variants use the selected and enriched data-driven ranks, respectively. The dashed red line denotes the true causal effect.}
  \label{fig:sim_nonlinear}
\end{figure}

Figures~\ref{fig:sim_nongaussian} and~\ref{fig:sim_nonlinear} summarize the results. Across the non-Gaussian settings, FC(selected) is nearly unbiased, and FC(enriched) is similarly centered with a modest increase in dispersion in some settings. IFE is generally intermediate and NUC is consistently upward biased. Under nonlinear misspecification, the impact depends on where the nonlinearity enters. When the outcome-confounding term is nonlinear but the exposure process remains linear in the latent factors, both FC variants remain close to the truth, suggesting that the relevant confounding structure can still be recovered from the exposures. In contrast, when the exposure model itself is nonlinear, the observed exposures are only approximately low-rank in a linear factor representation. In the pure cubic exposure setting, the parsimonious selector chooses rank \(3\) in all replications and shows a small residual downward bias; the enriched rule fits rank \(6\) and largely removes this bias, with only a modest variance increase. Overall, these results suggest that the proposed estimator is robust to a range of distributional departures and to nonlinear outcome confounding, while emphasizing that approximate nonlinear exposure confounding benefits from rank-enriched sensitivity analysis.


\subsection*{Neighborhood Interference}
\label{sec:sim_nearest_neighbor}

Our identification strategy accommodates partial interference, provided suitable negative controls are available. To illustrate this, we examine spatial neighborhood interference, where the outcome at each location is influenced not only by the local exposure but also by the exposures of its three nearest spatial neighbors. Specifically, we assume the following linear model:
\begin{equation}
Y_{it} = \beta_0 + \beta_1D_{it} + \beta_2\bar D_{it}+ \gamma_i^{\top}\Sigma_{U\mid D}^{-1/2} \mathbf U_{t} + \varepsilon_{it} \label{eqn:neighborhood_interference}
\end{equation}
where $\bar{D}_{it}$ denotes the average exposure from the three nearest neighbors, $\beta_1$ captures the direct local effect, and $\beta_2$ is the indirect spillover effect. We generate $100$ spatiotemporal datasets, assuming $N=50$ spatial locations, $T=200$ time periods, $M=4$ unmeasured confounders, and true effects $\beta_1 = 1$ and $\beta_2 = 0.5$. Here, we compare three methods for estimation: (i) standard DML without unmeasured confounder adjustment, (ii) the IFE estimator, and (iii) the proposed factor confounding estimator. Figure~\ref{fig:sim_interference} shows that both naive DML and IFE are biased in estimating $\beta_1$ and $\beta_2$. In contrast, the factor confounding approach provides nearly unbiased estimates of both effects, while also achieving lower variance than IFE.

\begin{figure}[htbp]
  \centering
  \includegraphics[width=0.8\linewidth]{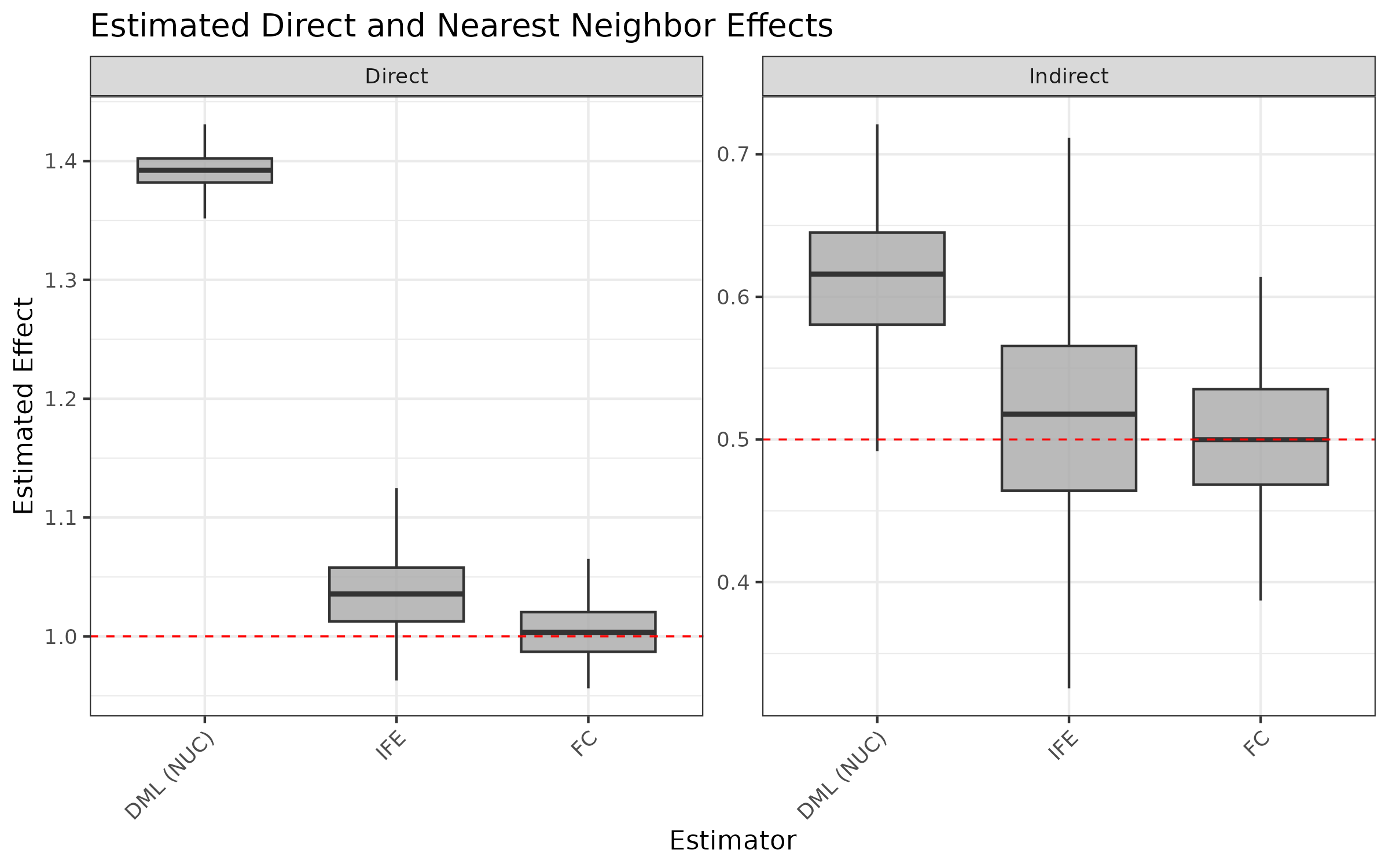}
  \caption{Sampling distributions of estimated direct and spillover effects under neighborhood interference from model \eqref{eqn:neighborhood_interference}. Dashed red lines denote the true effects.}
  \label{fig:sim_interference}
\end{figure}

\subsection*{Nonlinear Heterogeneous Effects}
\label{sec:appendix_nonlinear_sim}

We next consider a model with heterogeneous and nonlinear effects. We set $T = 1000,\; N = 6,\; M = 2$ and $p = 2$. The observed covariates and latent factors are generated as $\mathbf X_t \sim \mathcal N_p(0,I_p)$ and $\mathbf U_t \sim \mathcal N_M(0,I_M)$, and the residual terms satisfy $\bm \xi_t \sim \mathcal N_N\bigl(0,\sigma_\xi^2 I_N\bigr)$ and $\bm \varepsilon_t \sim \mathcal N_N\bigl(0,\sigma_\varepsilon^2 I_N\bigr)$, with $\sigma_\xi = \sigma_\varepsilon = 0.5.$ The parameter matrices are given by
\[\
\alpha = 
  \begin{pmatrix}
   1.0 & -1.2\\[2pt] 0.8 &  0.5\\[2pt] -0.4 &  1.0\\[2pt] 0.6 & -0.7\\[2pt] 1.5 &  0.9
  \end{pmatrix},\qquad
B =
  \begin{pmatrix}
    1.0 &  0.3\\ 0.4 & -0.8\\ -0.6 & 1.0\\ -0.7 & -0.3\\ 1.0 & -0.5
  \end{pmatrix},\qquad
\tilde\Gamma =
  \begin{pmatrix}
    0.4 & -0.7\\ -0.3 & 0.2\\ 0.8 & 0.2\\ 0.2 & 0.7\\ 0.5 & 0.4
  \end{pmatrix}.                            
\]
The exposures and outcomes are generated as   
\begin{align*}
    D_{it}  \;&=\; \alpha_{i1}\,\sin X_{1t} \;+\; \alpha_{i2}\,(X_{2t}^{2}-1) \;+\; B_{i\cdot} \mathbf U_t \;+\; \xi_{it},\\
    Y_{it} \;&=\; g_i(D_{it},X_{1t},X_{2t}) \;+\; \tilde\Gamma_{i\cdot} \mathbf U_t \;+\; \varepsilon_{it},
\end{align*}
Thus, both the exposure and outcome models are nonlinear, and the treatment effect varies across outcomes through the functions $g_i$. Figure~\ref{fig:nonlinear_simulation} compares the estimated and true causal effect curves evaluated at the mean of $\mathbf X$. The naive estimator is positively biased for outcomes $1$ and $5$ and negatively biased for outcomes $2$, $3$, and $4$, whereas the proposed factor confounding approach successfully removes these biases.  


\begin{figure}[htbp]
  \centering
  \includegraphics[width=0.8\linewidth]{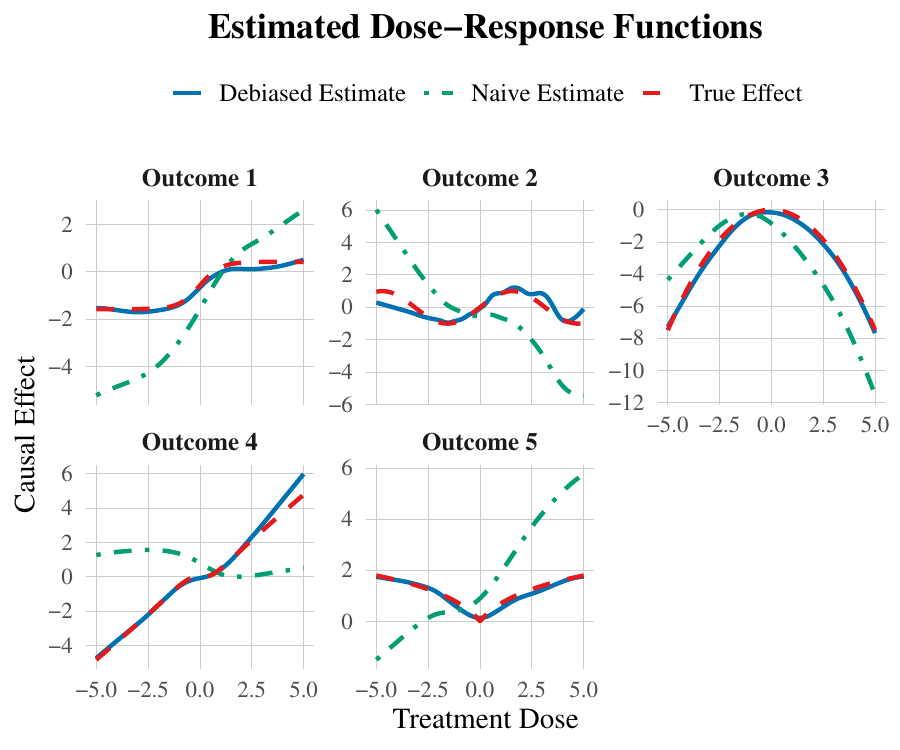}
  \caption{Comparison of true and estimated causal effect curves for each outcome under the nonlinear heterogeneous-effects setting, evaluated at the mean of observed covariates.}
  \label{fig:nonlinear_simulation}
\end{figure}


\clearpage

\section{Additional Figures for the Birth Weight Application}

\begin{figure}[htbp]
  \centering
  \includegraphics[width=0.8\linewidth]{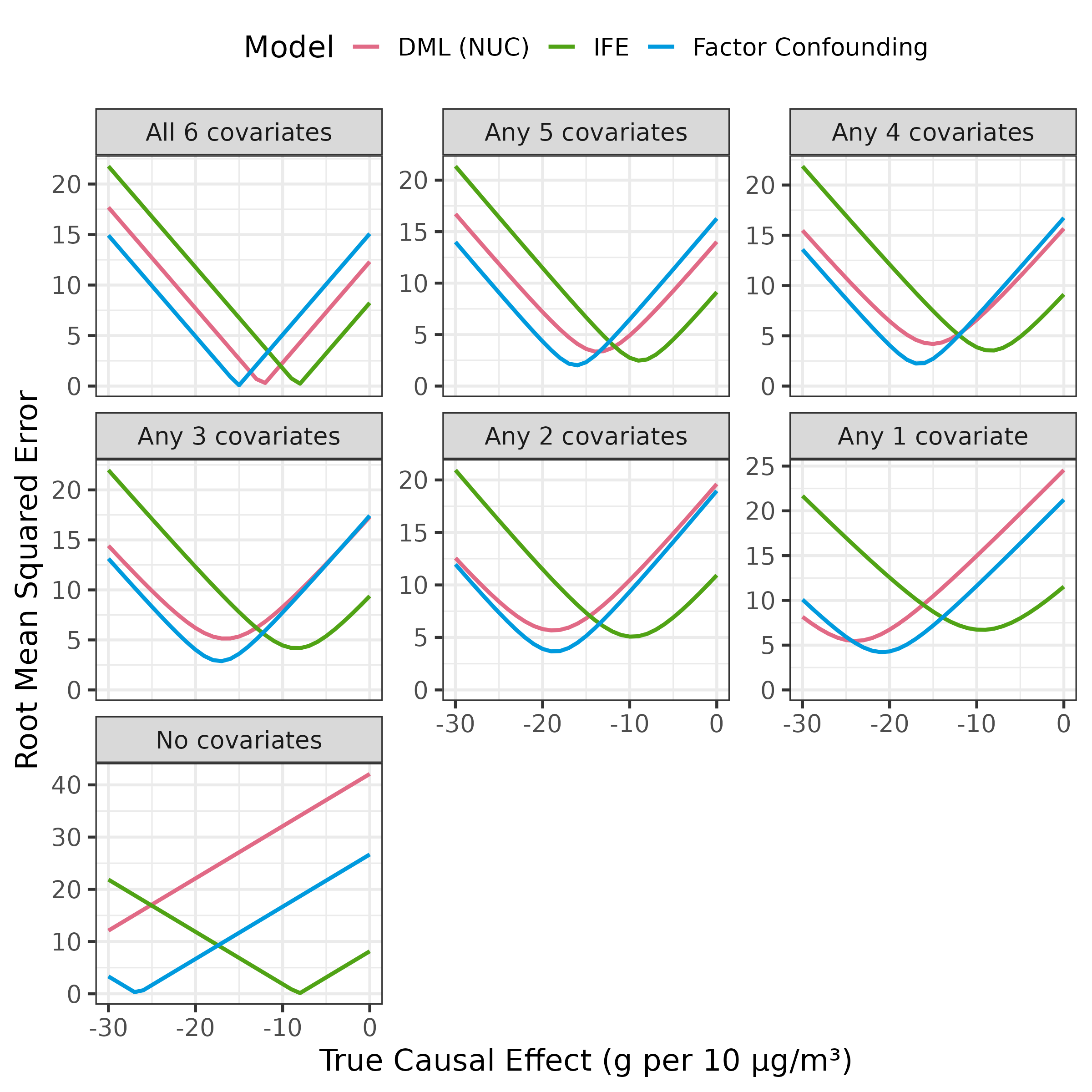}
  \caption{True causal effect versus RMSE of the estimates derived by adjusting for $\binom{6}{k}$ covariate groups, colored by model and faceted by $k$.  Factor confounding (blue) tends to have the lowest RMSE when the true causal effect is less than about $-15$ $g$, whereas IFE (green) performs best when the effect is larger than about $-12$ $g$.  Previous meta-analyses estimate the effect of a $10 \mu g / m^3$ increase in \pmtf on birth weight at around $-16$ $g$ \citep{gong2022maternal, gilbert2021causal}, or even greater reductions in the range of $-22$ $g$ to $-28$ $g$ \citep{ghosh2021ambient, uwak2021application}.} 
  \label{fig:sensitivity_rmse}
\end{figure}

\begin{figure}[htbp]
  \centering
  \includegraphics[width=0.7\linewidth]{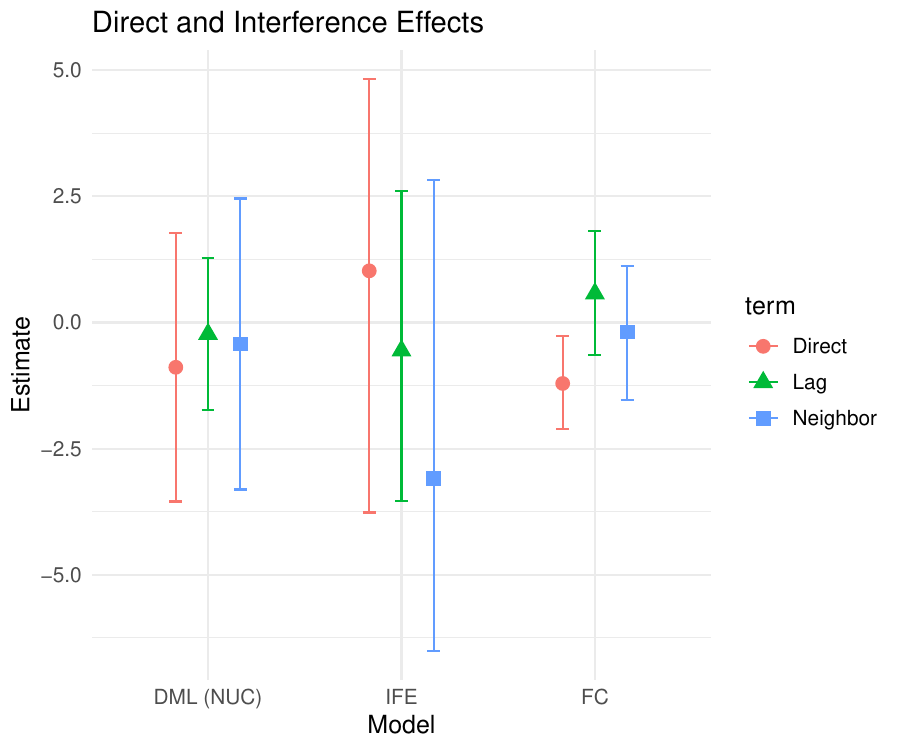}
  \caption{Estimated direct, lag, and neighborhood effects under DML (NUC), IFE, and factor confounding (FC) estimators. The factor confounding method recovers a significant direct effect with greater stability, while interference effects remain weak across all approaches. \label{fig:bw_interference}}
\end{figure}

\begin{figure}[htbp]
  \centering
  \includegraphics[width=0.9\linewidth]{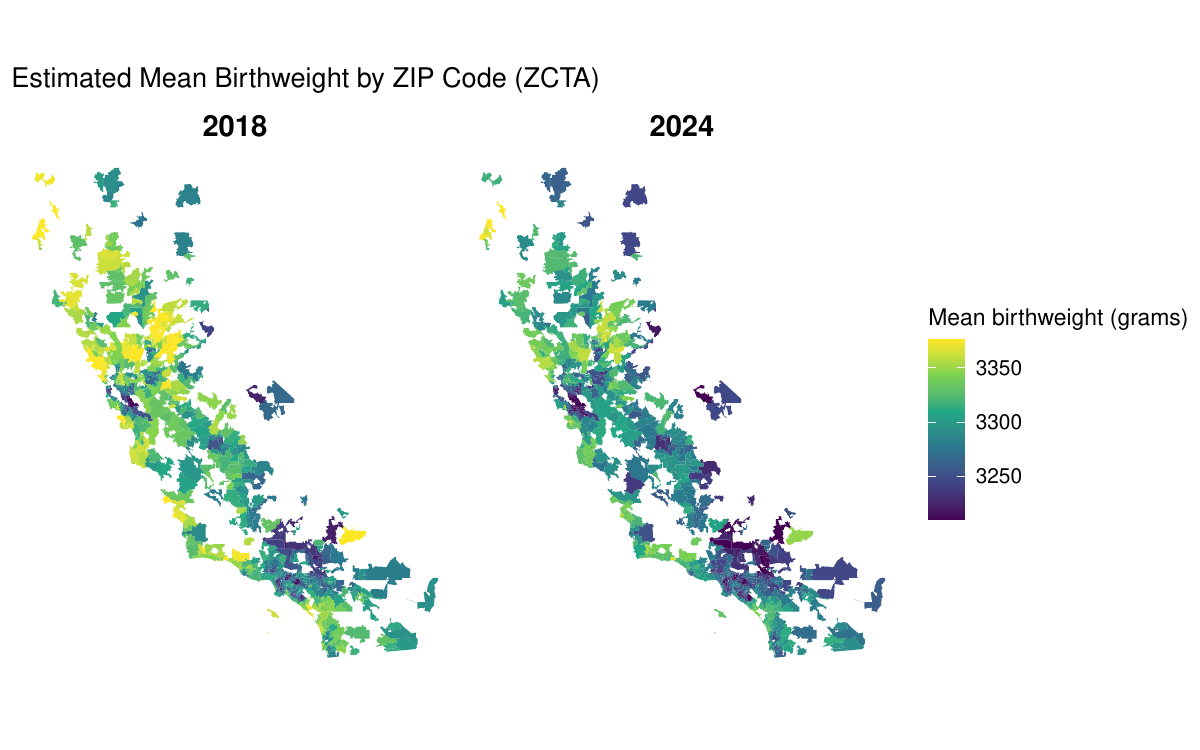}
  \caption{Estimated mean birth weight surface at the ZIP-code level for 2018 and 2024, based on adjustment for observed covariates, year indicators, and a smooth spatial surface in latitude and longitude. The fitted mean varies across both space and time, illustrating that broad spatiotemporal structure is modeled in the observed mean rather than being forced into the residual latent covariance.\label{fig:bw_mean_map}}
\end{figure}

\end{document}